\documentclass[lettersize,journal]{IEEEtran}
\usepackage{amsmath,amsfonts,amssymb}
\usepackage{mathtools}
\usepackage{algorithmic}
\usepackage{algorithm}
\usepackage{amsthm}
\usepackage{array}
\usepackage[caption=false,font=normalsize,labelfont=sf,textfont=sf]{subfig}
\usepackage{enumitem}
\usepackage[normalem]{ulem}
\usepackage{textcomp}
\usepackage{stfloats}
\usepackage{adjustbox}
\usepackage{booktabs}
\usepackage{url}
\usepackage{verbatim}
\usepackage{graphicx}
\usepackage{xcolor}
\usepackage{cite}
\usepackage{soul}
\usepackage{stmaryrd}

\newtheorem{theorem}{Theorem}
\newtheorem{lemma}{Lemma}

\theoremstyle{definition}
\newtheorem{definition}{Definition}

\begin{document}

\title{Greedy Selection for Heterogeneous Sensors}

\author{Kaushani Majumder,~\IEEEmembership{Student Member,~IEEE}, Sibi~Raj~B.~Pillai, Satish Mulleti,~\IEEEmembership{Member,~IEEE}
\thanks{K. Majumder, S. R. B. Pillai, and S. Mulleti are with the Department of Electrical Engineering, Indian Institute of Technology Bombay, Mumbai, 400076, India. Emails: kaushanim@iitb.ac.in, bsrajs@gmail.com, mulleti.satish@gmail.com}
}

\markboth{Submitted to the IEEE Transactions on Signal Processing}%
{Shell \MakeLowercase{\textit{et al.}}: Greedy Algorithm for Sensor Selection in Heterogeneous Sensor Networks}


\maketitle

\renewcommand\qedsymbol{$\blacksquare$}

\begin{abstract}
Simultaneous operation of all sensors in a large-scale sensor network is power-consuming and computationally expensive. Hence, it is desirable to select fewer sensors. A greedy algorithm is widely used for sensor selection in homogeneous networks with a theoretical worst-case performance of $(1-1/e) \approx 63\%$ of the optimal performance when optimizing submodular metrics. For heterogeneous sensor networks (HSNs) comprising multiple sets of sensors, most of the existing sensor selection methods optimize the performance constrained by a budget on the total value of the selected sensors. However, in many applications, the number of sensors to select from each set is known apriori and solutions are not well-explored. For this problem, we propose a joint greedy heterogeneous sensor selection algorithm. Theoretically, we show that the worst-case performance of the proposed algorithm is bounded to $50\%$ of the optimum for submodular cost metrics. In the special case of HSNs with two sensor networks, the performance guarantee can be improved to $63\%$ when the number of sensors to select from one set is much smaller than the other. To validate our results experimentally, we propose a submodular metric based on the frame potential measure that considers both the correlation among the sensor measurements and their heterogeneity. We prove theoretical bounds for the mean squared error of the solution when this performance metric is used. We validate our results through simulation experiments considering both linear and non-linear measurement models corrupted by additive noise and quantization errors. Our experiments show that the proposed algorithm results in $4-10$ dB lower error than existing methods.
\end{abstract}


\section{Introduction}

In many applications such as medical imaging and healthcare \cite{constantinos2001medical, ko2010wireless, alemdar2010wireless}, seismic processing \cite{liu2013volcanic, savazzi2013ultra, khedo2020inland}, environmental monitoring \cite{hart2006environmental, rajasegarar2013high}, power networks \cite{bose1987real, terzija2010wide}, smart homes and internet of things (IoT) networks \cite{alam2012review, li2014design}, etc., a large number of sensors are deployed.
In these applications, the task is often to estimate a low-dimensional parameter vector from the sensor measurements. The accuracy of the estimation task increases with the number of sensors or measurements. However, this results in high power requirements and higher communication and storage costs. To reduce operational costs, one can deploy fewer sensors. 
Such fixed deployment is inefficient and non-adaptive when the measuring fields and their noise characteristics change over time. Alternatively, deploying a larger number of sensors and then choosing the desired numbers based on the information available apriori is advantageous. Activating only a few sensors at a time also leads to a considerable reduction in operating costs.

\begin{figure}[t]
    \centering
    \includegraphics[width= 2.6 in]{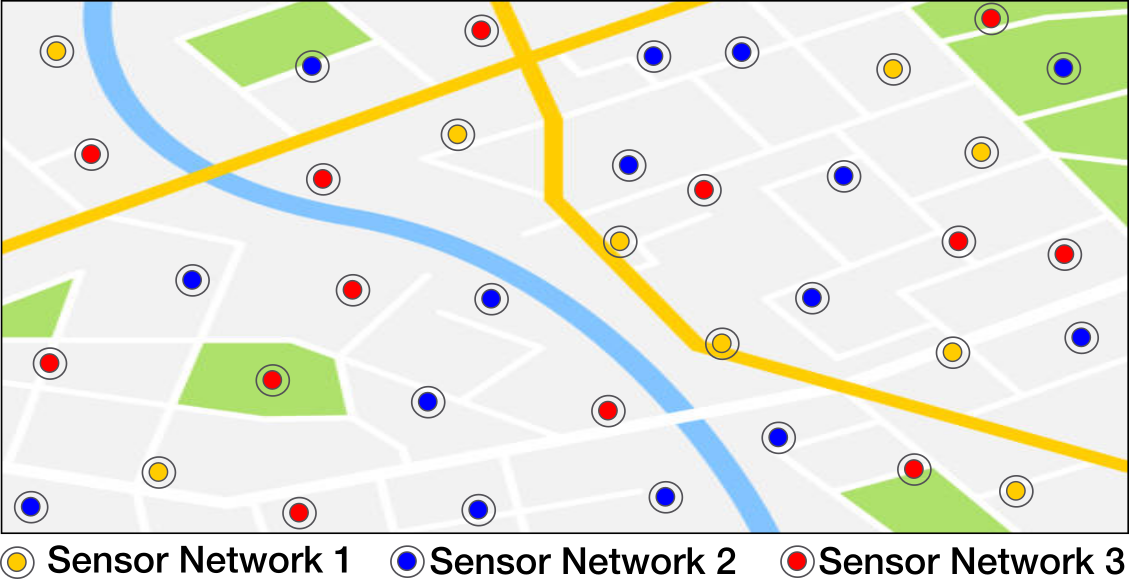}
    \vspace{-0.1in}
    \caption{
    An example of an HSN comprising three subnetworks of sensors. All the sensors observe the same phenomenon.}
    \label{fig:HSN}
    \vspace{-0.2in}
\end{figure}

The sensor networks could be either homogeneous, where all the sensors have similar accuracies, or heterogeneous. In heterogeneous sensor networks (HSNs), multiple sets of sensors observe the same or diverse phenomena \cite{viani2013wireless, peters2015utilize, rajasegarar2013high}. When an HSN measures a common phenomenon, the sensors in different sets are characterized by their accuracies or, equivalently, their noise levels. Such an HSN is shown in Fig.~\ref{fig:HSN} \cite{peters2015utilize, rajasegarar2013high}. 
For example, consider an environmental monitoring task where multiple sets of sensors with different precisions or noise characteristics are deployed to observe a phenomenon \cite{hart2006environmental, rajasegarar2013high}. Data fusion is performed on the measurements from all these sets in order to obtain a more robust estimate of the phenomenon under observation. To reduce the communication and processing cost and extend each subnetwork's lifetime, only a small number
of observations from each set needs to be selected to achieve the best possible estimation accuracy. 

The sensor selection problem in HSNs can also be used for sensor placement problems when the locations for the different sensor sets are predefined. For example, in weather monitoring networks \cite{zhang2018spatial}, high-precision weather stations can only be placed at sparse locations over a geographical region due to their high cost of construction, unavailability of suitable installation locations, or high power consumption. The estimation inaccuracy induced by the low spatial density of these high-precision sensors is compensated by deploying in bulk, low-powered, inexpensive, low-precision sensors. The question then arises: given probable placement locations for the different sets of sensors, which locations should be selected to ensure the best possible accuracy in estimating weather conditions?

In both these examples and several such practical scenarios, the task is to maximize a performance criterion constrained by the cardinality of sensors to be selected from each subset. In a few other applications, the sensor selection problem is formulated with a budget constraint, assigning each sensor a predetermined value \cite{sviridenko2004note, sviridenko_2017_optimal, badanidiyuru2014fast}. However, in the examples considered, cardinality constraints describe the problem better, resulting in a combinatorial optimization problem for which finding the optimal solution is computationally expensive. Even in homogeneous sensor networks, selecting a subset of sensors to reduce the operational cost is an NP-hard problem \cite{joshi2008sensor}. 
Various methods have been proposed to find an approximate solution to this problem, such as convex relaxation \cite{joshi2008sensor, chepuri2014sparsity}, different heuristics \cite{chiu2004simulated, al2008optimizing, mukherjee2006systematic, mackay1992information, wang2004entropy}, and a greedy selection procedure \cite{nemhauser1978analysis, Mirzasoleiman_2015_Lazier, hashemi_2020_randomized}. While theoretical performance bounds are available for the convex relaxation methods, they are computationally expensive. Heuristics perform well in practice, although no theoretical guarantee is available for their performance. The greedy algorithm, on the other hand, is computationally faster than the convex relaxation approach while achieving a worst-case performance of $\left( 1 - 1/e \right) \approx 63\%$ of the optimal performance when optimizing a submodular cost metric \cite{nemhauser1978analysis}, where $e$ is Euler's constant.

On the other hand, limited results are available for the sensor selection problem in HSNs. In \cite{wolsey_1982_analysis, sviridenko2004note, lin_2010_multi, badanidiyuru2014fast, sviridenko_2017_optimal, esmaeilbeig_2023_submodular}, the authors develop a modified greedy algorithm that maximizes the performance while satisfying a budget constraint on the total operating cost. These methods are shown to have performance guarantees approaching the limit $\left( 1 - 1/e \right)$. However, they cannot be used when the number of sensors to choose from each subset is specified. Zhang et al. \cite{zhang2018spatial} proposed a cross-entropy-based stochastic optimization method for the HSN sensor selection problem where they considered minimizing the total operational cost of the selected sensors while achieving a specified mean squared error (MSE). In this approach, there is an undesirable possibility that the algorithm may choose all the available sensors to achieve a given MSE. 

Kirchner et al. \cite{kirchner2020heterogeneous} consider the heterogeneous sensor selection problem in vehicle tracking applications, where sets of sensors obtain multiple snapshots of a common phenomenon. The task is to find the optimal number of snapshots each sensor needs to communicate with a central data processing unit, with a budget on the maximum bandwidth allocated to the sensor. Their approach is indeed to independently and greedily select measurements from each sensor without considering the inter-sensor correlations. Moreover, \cite{kirchner2020heterogeneous} does not reduce the number of sensors.

The problem of sensor placement in HSNs has been considered in \cite{zhang2017sensor, bushnaq2020sensor}. In \cite{zhang2017sensor}, high or low-power sensor placement is considered at each available location based on the feedback of successful transmission. On the other hand, \cite{bushnaq2020sensor} considers sensor placement in IoT networks, optimizing the energy harvesting capabilities of the diverse sensors according to  their channel gains and the required accuracy of the estimation task. These methods fundamentally differ from the sensor selection task since the type of sensor at each location is not fixed, and moreover, the number of sensors of each type to select is also a design parameter in \cite{zhang2017sensor, bushnaq2020sensor}. 

All the state-of-the-art methods for sensor selection in HSNs suffer from the drawback that they do not work under cardinality constraints. 
In this paper, we propose an algorithm to select the optimal sensors when the number of sensors to be selected from each set is specified. Each sensor measures a spatial observation induced by a set of deterministic unknown parameters, which are to be estimated. Different noise distributions characterize each set of sensors. In such a scenario, finding the optimal subset from each set is considered while achieving the maximum possible performance.

The main contributions of this paper are as follows.
\begin{itemize}
    \item This work introduces a greedy algorithm for selecting sensors in HSNs under cardinality constraints specifying the number of sensors to select from each subset.

    \item Theoretically, the proposed algorithm is proved to achieve at least $50\%$ of the optimal performance when optimizing a submodular cost metric. In the special case of HSNs comprising two sets of sensors, the theoretical performance guarantee can be improved to $(1 -1/e)$ when the number of high-precision sensors to be chosen is less than a fraction ($10\%$ or lower) of the low-precision ones.

    \item A metric called weighted frame potential (WFP) is introduced, which extends the frame potential (FP) criterion for sensor selection to HSNs. It exploits both the correlation between the measurements and their noise characteristics. We prove optimality guarantees for WFP and related guarantees for the mean squared error of the solution. A submodular performance metric called weighted frame cost (WFC) based on WFP is proposed.

    \item Extensive experiments are conducted to evaluate the proposed algorithm for both linear as well as non-linear measurement models. We perform small-scale experiments to compare the performance of the proposed method with the optimal solution obtained through an exhaustive search. Large-scale sensor selection problems are simulated and the proposed method is compared against other greedy-based and random approaches. Our experiments consider the effect of both additive Gaussian noise and quantization on the measurements. Finally, as an application of the proposed algorithm in non-linear measurement models, we consider the problem of selecting sensors for estimating the direction-of-arrival of plane waves emitted by multiple sources.

    \item The small-scale experiments show that the proposed algorithm achieves near-optimal solutions (over $99\%$ of the optimal WFC), which is better than the theoretical performance guarantees obtained. This is similar to the performance of existing greedy methods for sensor selection in homogeneous networks \cite{nemhauser1978analysis}.

    \item In all the experiments conducted, the proposed method consistently achieves a mean squared error of $4$-$10$ dB less than the existing methods.

\end{itemize}

The rest of the paper is organized as follows. In Section~\ref{sec:prob_def}, we define the problem. The greedy algorithm for homogeneous sensor selection is presented in Section~\ref{sec:Preliminaries}. The proposed algorithm is introduced in Section~\ref{sec:JGS}, and its theoretical performance guarantees are derived. The WFP metric and the related optimality bounds are presented in Section~\ref{sec:WFP}. Various experiments and simulation results are presented in Section~\ref{sec:Experiments}, and Section~\ref{sec:Conclusion} concludes the paper.

Throughout the paper, vectors (matrices) are represented as bold-faced lower-case (upper-case) letters such as $\mathbf{a}$ ($\mathbf{A}$). 
We use $\mathbf x \in \mathbb R^N$ ($\mathbf x \in \mathbb C^N$) to denote a real (complex) $N$-dimensional vector. The norm of a vector $\mathbf{x}$ belonging to a normed space is denoted as $\| \mathbf{x} \|$, and the inner product between two vectors $\mathbf{x}$ and $\mathbf{y}$ is denoted as $\langle \mathbf{x}, \mathbf{y} \rangle$. In particular, $\| \mathbf{x} \|_2$ denotes the Euclidean norm of $\mathbf{x}$. Sets are represented by upper-case calligraphic letters like $\mathcal{S}$, with $|\mathcal{S}|$ denoting its cardinality and $\mathcal{P} \left( \mathcal{S} \right)$ standing for the power set of $\mathcal S$. A set function $\mathcal{C}: \mathcal{P} \left( \mathcal{S} \right) \rightarrow \mathbb{R}$ is defined as a function that acts on any subset $\mathcal T \subseteq \mathcal{S}$, producing a real-valued output. For a positive integer $N$, we use the following representation $\mathcal{N} =\{ 1, 2, \ldots, N \}$.

\section{Problem Definition} \label{sec:prob_def}

Consider the problem of recovering an unknown vector $\mathbf{x} \in \mathbb{C}^K$ from its noisy measurements obtained by $N \geq K$ sensors. The measurements are modeled as
\begin{align}
    \mathbf{y} = \mathcal{A}(\mathbf{x})  +\boldsymbol{\eta} ,
    \label{eq: problem_model}
\end{align}
where $\mathbf{y}\in  \mathbb{C}^N$is the measurement vector with its $i$-th element representing observation from the $i$-th sensor. The measurement operator $\mathcal{A}: \mathbb{C}^{K} \rightarrow \mathbb{C}^N$ is a function of the locations of the sensors and $\boldsymbol{\eta} \in \mathbb{C}^N$ is an additive noise vector. The measurement operator could be linear or non-linear. While $\mathcal{A}(\cdot)$ is represented as a measurement matrix $\mathbf{A} \in \mathbb{C}^{N \times K}$ in the linear case,
yielding
\begin{align}
    \mathbf{y} = \mathbf{A}\mathbf{x}  +\boldsymbol{\eta},
    \label{eq: problem_model_linear}
\end{align}
more general non-linear models also come under our purview.
For example, in the direction-of-arrival (DoA) estimation problem, $\mathcal{A}(\mathbf{x})$ is given as $\mathbf{A}(\boldsymbol{\theta}) \boldsymbol{\alpha}$ where $\boldsymbol{\theta}$ and $\boldsymbol{\alpha}$ denote angles and amplitudes of sources. In this case, the unknown parameters to be estimated are $\mathbf{x} = [\boldsymbol{\theta}^{\mathrm{T}} \,\, \boldsymbol{\alpha}^{\mathrm{T}}]^{\mathrm{T}}$. Here, since $\mathbf{A}(\boldsymbol{\theta})$ is a function of the parameters $\boldsymbol{\theta}$, the observations $\mathbf{y}$ are not linear functions of $\boldsymbol{\theta}$.

In our work, we consider a heterogeneous sensor network comprising $L$ sets of sensors $\mathcal{S}_1, \mathcal{S}_2, \ldots, \mathcal{S}_L$, with  $\mathcal{S}_i\subseteq \mathcal{N}$, for $i = 1, \ldots, L$, where $\displaystyle \bigcup_{i=1}^L \mathcal{S}_i = \mathcal{N}$ and $\displaystyle \bigcap_{i=1}^L \mathcal{S}_i = \varnothing$. We assume that the sensors in the different sets have varying precision. To model the precision, we assume that the elements of the noise vector $\boldsymbol{\eta}$ corresponding to the set $\mathcal{S}_i$ have standard deviation $\sigma_i$, $i = 1, \ldots, L$. We further assume that elements of $\boldsymbol{\eta}$ for two different sensors are independent. Hence, the covariance matrix of the noise vector is given as $\boldsymbol{\Sigma} = \text{diag} \left( \sigma_1^2, \sigma_2^2, \ldots, \sigma_N^2 \right)$, with $\sigma_j^2 = \sigma_i^2$ if $j \in \mathcal{S}_i$, $i = 1, \ldots, L$.

Our objective is to select $M_i$ sensors from $\mathcal{S}_i$, for $i = 1, \ldots, L$, such that a given accuracy metric $\mathcal{C}$ is maximized. This is captured in the following optimization problem.
\begin{align}
\begin{aligned}
    \underset{\mathcal{T}_i \subseteq \mathcal{S}_i, i = 1, \ldots, L}{\arg \max} \quad  &\mathcal{C}\left( \bigcup_{i=1}^L \mathcal{T}_i \right), \\
    \text{subject to } &\left| \mathcal{T}_i \right| = M_i, \quad \text{for } i = 1,\ldots, L \notag
    \end{aligned} \tag{P1} \label{eq:relaxed_prob}
\end{align}
Notice that \eqref{eq:relaxed_prob} is a combinatorial optimization problem for which no polynomial time algorithms are known to exist. Particularly, finding the optimal solution requires a computationally expensive search over all possible subsets of $\mathcal{S}_1, \ldots, \mathcal{S}_L$, respecting the cardinality constraints. Instead, we propose an algorithmic solution to (\ref{eq:relaxed_prob}) using a greedy approach and characterize its worst-case performance theoretically.  In order to lay the foundations for
our algorithm, we first present the well-known greedy selection (GS) algorithm for homogeneous sensor networks ~\cite{nemhauser1978analysis}.

\section{Greedy Selection For Homogeneous Sensors} \label{sec:Preliminaries}
The problem of sensor selection has been well explored in the case of homogeneous sensor networks composed of sensors with similar noise levels \cite{nemhauser1978analysis, joshi2008sensor, chepuri2014sparsity, chiu2004simulated, al2008optimizing, mukherjee2006systematic, mackay1992information, wang2004entropy}. 
Consider the problem of selecting $M$ out of $N$ sensors to maximize a metric $\mathcal{C}$. The corresponding optimization problem is formulated as
\begin{align}
    \underset{\mathcal{T} \subseteq \mathcal{S}}{\arg \max}\quad \mathcal{C}\left( \mathcal{T} \right), \label{eq:sensor_sel_gen}
    \quad\text{subject to }  \left| \mathcal{T} \right| \le M. \tag{P2} 
\end{align}
Here $\mathcal{S}$ denotes the set from which the required subset $\mathcal{T}$ of size at most $M$ needs to be selected.
 
An iterative greedy selection algorithm to solve (\ref{eq:sensor_sel_gen}) was proposed by Nemhauser et al. \cite{nemhauser1978analysis}. This method starts with an empty set $\mathcal{T}$. At each iteration, the algorithm chooses a new measurement/sensor from the remaining ones, which maximizes the performance metric $\mathcal{C}$ (or alternatively minimizes estimation error). This step is repeated $M$-times to select those many measurements (See Algorithm \ref{alg:greedy_alg_gen}).
\begin{algorithm}[tb]
    \caption{Greedy Algorithm for Homogeneous Sensor Selection}
    \label{alg:greedy_alg_gen}
    \begin{algorithmic}[1]
        \STATE \textbf{Initialize:} $\mathcal{T} = \varnothing$
        \FOR{$m = 1 \text{ to } M$}
            \STATE $i^* = \underset{i \notin \mathcal{T}}{\arg \max \text{ }} \mathcal{C} \left( \mathcal{T} \cup  \{i\}  \right)$ \COMMENT{Greedy Search}
            \STATE $\mathcal{T} \gets \mathcal{T} \cup \{i^*\}$
        \ENDFOR
    \end{algorithmic}
\end{algorithm}

Nemhauser et al. \cite{nemhauser1978analysis} showed that when $\mathcal{C}$ is a normalized monotone non-decreasing submodular function, then the worst case error in the performance of the solution given by the GS algorithm is bounded as follows,
\begin{align}
    \frac{\mathcal{C}(\mathcal{T}^\ast)}{\mathcal{C} \left( \mathcal{T}_{OPT} \right) } \ge \left( 1 - \left( \frac{M-1}{M} \right)^{^M} \right) 
    \ge \left( 1 - \frac{1}{e} \right),
    \label{eq:greedy_selection}
\end{align}
where $\mathcal{T}^\ast$ is the solution found by the greedy algorithm and $\mathcal{T}_{OPT}$ is the optimal solution of~\eqref{eq:sensor_sel_gen}. Let  $\mathcal{P}\left( \mathcal{N} \right)$ denote the power set of $\mathcal{N}$. Then a set function $\mathcal{C} : \mathcal{P}\left( \mathcal{N} \right) \rightarrow \mathbb{R}$ is said to be monotone increasing if $\mathcal{C}(\mathcal{S}) \ge \mathcal{C}(\mathcal{T})$, $\forall \mathcal{T} \subseteq \mathcal{S} \subseteq \mathcal{N}$, and is said to be normalized if $\mathcal{C}(\varnothing) = 0$, where $\varnothing$ denotes the empty set. A submodular function is a class of set functions that follows the property of diminishing returns, mathematically stated as $\mathcal{C}\left( \mathcal{S} \cup \{j\} \right) - \mathcal{C} \left( \mathcal{S} \right) \le \mathcal{C}\left( \mathcal{T} \cup \{j\} \right) - \mathcal{C} \left( \mathcal{T} \right), \forall \mathcal T \subseteq \mathcal S$.
The performance guarantee stated in (\ref{eq:greedy_selection}) theoretically bounds the GS algorithm's worst-case performance to $\left( 1 - 1/e \right) \approx 63\%$ of the optimal performance.

In the next section, we propose a modified version of the greedy algorithm to make it suitable for selecting sensors in HSNs, as described in \eqref{eq:relaxed_prob}, and analyze its performance.

\section{A Joint Greedy Selection Algorithm for Heterogeneous Sensor Networks} \label{sec:JGS}
The problem of sensor selection in HSNs, as stated in (\ref{eq:relaxed_prob}), has not been explored much in the literature, to the best of our knowledge. A straightforward method is to choose $M_i$ sensors from $\mathcal S_i$, for $i=1,\ldots, L$, independently using the homogeneous greedy selection algorithm. We call this the independent greedy selection (IGS) approach. Another technique respecting the constraints is to independently and randomly select (IRS) the required number of sensors from each set. Last but not least, an exhaustive search can find the optimal answer but has prohibitive computational complexity. For comparison, let $\mathcal T_{OPT}$ be an optimal set identified by exhaustive search. While it is clear that the first two approaches quickly identify feasible solutions, it will be unraveled later that they fail to yield good performance. Therefore, as discussed next, we turn to a joint greedy selection (JGS) approach.

\subsection{Proposed Joint Greedy Selection Algorithm}
 The JGS algorithm simultaneously considers all the sets $\mathcal{S}_1, \ldots, \mathcal{S}_L,$ to circumvent the shortcomings of the independent selection approaches discussed. The algorithm adds one sample at a time, starting from $L$ empty sets. Specifically, we start with $\mathcal{T}_i = \varnothing, i = 1, \ldots, L$. Let $\displaystyle \mathcal{T} = \bigcup_{i=1}^L \mathcal{T}_i$ denote the set of selected sensors. Then for $|\mathcal{T}_i|<M_i$, $i = 1, \ldots, L$, at every iteration, we select a new sample $t^\ast \notin \mathcal{T}$ such that the cost $\mathcal{C}(\mathcal{T} \cup \{ t^\ast \})$ is maximized. After the selection, the set $\mathcal{T}_i$ is updated depending on whether the sample $t^\ast$ belongs to $\mathcal{S}_i$ for $i = 1, \ldots, L$. The iterations continue until one of the sets is exhausted, that is, until $\left| \mathcal{T}_i \right| = M_i$ condition is satisfied for $i \in \{1, \ldots, L\}$. After this point, the search is restricted to the sets where $\left| T_i \right| < M_i$. The iterations continue until we select the desired number of sensors from each set, exhausting one set at a time. These steps are summarized in Algorithm~\ref{alg:JGS}.

It is to be noted that in JGS, there is a change in the search space when the algorithm exhausts one of the subsets, that is, when any of $\left| \mathcal{T}_i \right| = M_i$ for $i = 1, \ldots, L$ is satisfied. In the first few iterations, JGS searches for the next candidate sensor from the entire set of unselected sensors. However, once one of the sets is exhausted, the algorithm can only search over the remaining subsets. This is the step where the algorithm essentially differs from the vanilla GS algorithm. It is unknown beforehand at which iteration this switch occurs, and as we discuss later, the iteration at which the switch takes place plays a key role in the performance analysis of the JGS algorithm.

\begin{algorithm}[tb]
    \caption{Joint Greedy Selection (JGS) Algorithm for Heterogeneous Sensor Networks with Submodular Cost}
    \label{alg:JGS}
    \begin{algorithmic}[1]
        \STATE \textbf{Initialize:} $\displaystyle \mathcal{T} = \varnothing, \mathcal{T}_i = \varnothing \text{ for } i = 1, \ldots, L , \mathcal{T}^c = \bigcup_{i = 1}^L \mathcal{S}_i$
        \FOR{$\displaystyle m = 1 \text{ to } \sum_{i = 1}^L M_i$}
            \STATE $t^\ast = \underset{t \in \mathcal{T}^c}{\arg\max} \quad \mathcal{C} \left( T \cup \{t\}  \right)$ 
            \COMMENT{Greedy Search} \\
            \COMMENT{Updating the currently selected sets} \\
            \IF{$t^\ast \in \mathcal{S}_i$, for $i \in \{1, \ldots, L\}$} 
                \STATE $\mathcal{T}_i \gets \mathcal{T}_i \cup \{t^\ast\}$
            \ENDIF
            \STATE $\displaystyle \mathcal{T} = \bigcup_{i=1}^L \mathcal{T}_i$ \\
            \COMMENT{Updating the set to select from} \\
            \STATE $\mathcal{T}^c = \varnothing$
            \FOR{$i = 1, \ldots, L$}
            \IF{$\left| \mathcal{T}_i \right| < M_i$}
                \STATE $\mathcal{T}^c \gets \mathcal{T}^c \cup \left( \mathcal{S}_i - \mathcal{T}_i \right)$
            \ENDIF
            \ENDFOR
        \ENDFOR
    \end{algorithmic}
\end{algorithm}

\subsection{Performance Guarantee}
The proposed JGS algorithm may generate a suboptimal solution for (\ref{eq:relaxed_prob}). Hence, to evaluate the algorithm, it is essential to determine its worst-case performance. In this section, we prove a theoretical bound for the maximum deviation of the achieved performance from the optimal performance. For the rest of this section, we shall assume that the cost function $\mathcal{C}$ being maximized by the algorithm is a normalized, monotone, non-decreasing submodular function.

Our main result, summarized in the following theorem, gives a bound on the worst-case performance of the proposed JGS algorithm in terms of the cost achieved by an optimal set $\mathcal{T}_{OPT}$ for (\ref{eq:relaxed_prob}).

\begin{theorem}\label{thm:Perf_Guar_Hetero_Half}
    Consider the optimization problem in (\ref{eq:relaxed_prob}) where $\mathcal{C}$ is assumed to be a normalized, monotone, non-decreasing, submodular function. Then the cost of the final set $\mathcal{T}^\ast$ selected by the algorithm is bounded as
    \begin{equation}
        \mathcal{C} \left( \mathcal{T}^\ast \right) \ge \frac{1}{2} \mathcal{C} \left( \mathcal{T}_{OPT} \right)
        \label{thrm_1}
    \end{equation}
\end{theorem}

The proof of this theorem is presented in Appendix~\ref{append:c}. Although Theorem~\ref{thm:Perf_Guar_Hetero_Half} gives a constant lower bound of achieving at least $50\%$ of the optimal performance, we can show tighter bounds when the number of sets is restricted to two sets $\mathcal{S}_1$ and $\mathcal{S}_2$, and the number of sensors with low noise characteristics is vastly outnumbered by the number of low-precision sensors with high noise. This result is presented in the next theorem.

\begin{theorem}\label{thm: Perf_Guar_Hetero}
    Consider the optimization problem in (\ref{eq:relaxed_prob}) where $\mathcal{C}$ is assumed to be a normalized monotone, non-decreasing, submodular function. Let the JGS algorithm 
    at the $m_s$-th iteration select $M_i$ elements from the set $\mathcal{S}_i$, for $i \in \{1,2\}$. Let $\mathcal{S}_{i'}$, where $i' \coloneqq \{1,2\} \backslash \{i\}$, denote the set complementary to $\mathcal{S}_i$. Then the cost of the final set $\mathcal{T}^\ast$ selected by the JGS algorithm is bounded as
    \begin{multline}
        \mathcal{C} \left( \mathcal{T}^\ast \right) \ge \left[1 - \frac{M_i}{M_{i'}} \sum_{j=0}^{M_i + M_{i'} - m_s - 1} \left( 1 - \frac{1 + M_i}{M_{i'}} \right)^j - \right. \\ 
        \left. \left( 1 - \frac{1 + M_i}{M_{i'}} \right)^{M_i + M_{i'} - m_s} \left( 1 - \frac{1}{M_i + M_{i'}} \right)^{m_s} \right] \mathcal{C} \left( \mathcal{T}_{OPT} \right).\\
        \label{thrm_1_1}
    \end{multline}
\end{theorem}
Proof of the theorem is deferred to Appendix~\ref{append:b}.

As seen from Theorem \ref{thm: Perf_Guar_Hetero}, 
The bounds in Theorem \ref{thm: Perf_Guar_Hetero} depend on $M_1$, $M_2$, and $m_s$. Thus, unlike the GS algorithm for homogeneous sensor networks, we cannot give a single number that specifies the worst-case error. Next, we discuss how the worst-case error varies with $M_1$, $M_2$, and $m_s$.

\subsection{Analysis of Performance Guarantee}
We assess the theoretical bounds in \eqref{thrm_1_1} by fixing $m_s$ and varying $M_1$ and $M_2$. For the sake of discussion, we assume that $M_1$ is reached first by the algorithm. Otherwise, $M_1$ and $M_2$ will be interchanged in the discussion and the following results. With $M_1$ samples selected first from set $\mathcal{S}_1$, we have that $M_1 \le m_s \le M_1 + M_2 - 1$.

\begin{figure}[t]
    \centering
    \includegraphics[width=2.5 in]{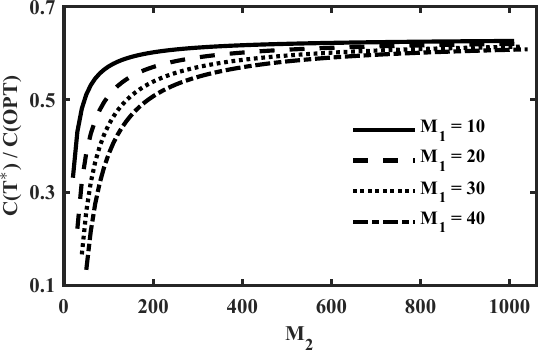}
    \caption{Plot of the theoretical worst case error for varying $M_1$ and $M_2$, when $m_s = M_1 + M_2 - 1$.}
    \label{fig:bound_interp_2}
    \vspace{-0.1in}
\end{figure}

Figure~\ref{fig:bound_interp_2} shows the performance guarantees when $M_1$ and $M_2$ vary, with $ m_s = M_1 + M_2 - 1$. We observe that the worst case performance asymptotically approaches the value $\left( 1 - 1/e \right) \approx 63\%$ as $M_1$ gets progressively smaller than $M_2$, with $\frac{M_1}{M_2} \rightarrow 0$.
Specifically, we observe that for $\frac{M_1}{M_2} = O\left( 10^{-2} \right)$ the worst case performance reaches $\approx 63\%$ of the optimal performance. For relatively larger values of $M_1$ compared to $M_2$, such as $\frac{M_1}{M_2} \approx 0.1$, the selected set has a performance which is at worst $55\%$ of the optimal performance.

\begin{figure}[t]
    \centering
    \includegraphics[width= 2.5 in]{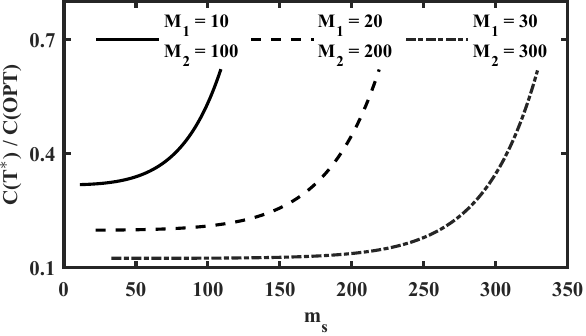}
    \caption{Plots of the theoretical worst case error guarantee for the proposed algorithm, when $m_s$ is varied from $M_1$ to $\left( M_1 + M_2 - 1 \right)$, for different values of $M_1$ and $M_2$.}
    \label{fig:bound_interp_1}
    \vspace{-0.1in}
\end{figure}

Figure~\ref{fig:bound_interp_1} shows the variation of the performance guarantee when $M_1$ and $M_2$ are kept constant and $m_s$ is varied. As the figure shows, the worst-case performance error for lower values of $m_s$ is not very encouraging. This is due to the approximations we make in deriving the bound and it may be improved if these approximations can be improved. However, lower values of $m_s$ are less likely since this scenario implies very little correlation among the measurements in $\mathcal{S}_1$ and $\mathcal{S}_2$. 
It is observed from Fig.~\ref{fig:bound_interp_1} that the performance remains flat for lower values of $m_s$ (for $m_s = M_1$ till about $50-60\%$ of $M_1 + M_2$), and then the performance increases as $m_s$ increases beyond this value until it reaches $(1 - 1/e)$ of the optimal performance. These results imply that the algorithm performs better when the switch in the search space occurs later in the algorithm. From Fig.~\ref{fig:bound_interp_2} and Fig.~\ref{fig:bound_interp_1}, we conclude that the theoretical performance guarantee obtained by the JGS algorithm is similar to the standard greedy method when $m_s$ takes higher values. Higher values of $m_s$ mean that JGS utilizes the correlation information between the two subsets for more number of steps. 

This analysis shows that while the JGS algorithm is guaranteed to produce a performance which is at worst $50\%$ of the optimum, this guarantee can be improved to $63\%$ under the special case when there are only two sets with the number of sensors to select from one set much smaller than the other. This situation practically arises when the two sets are composed of expensive high-precision sensors and easily available inexpensive low-precision sensors, respectively. In this case, the idea is to use a small number of high-precision sensors and compensate for the low sample density with an adequate number of low-precision sensors.

Next, we introduce a submodular performance metric for heterogeneous networks based on the frame potential. We obtain related performance guarantees for widely used error metrics such as the mean squared error.

\section{Performance Metrics and Submodular Surrogates} \label{sec:WFP}

In typical estimation tasks, MSE is employed as the performance measure. However, it is not submodular. Hence, theoretical guarantees derived in the previous section for greedy algorithms are not applicable. In practice, instead of MSE, its proxies such as Cram\'er-Rao lower bound (CRLB)-based (defined as the inverse of the Fisher information matrix (FIM) \cite{kay1993fundamentals}) or frame-potential (FP)-based \cite{ranieri2014near, Tohidi_2019_Sparse_antenna} functions are used with greedy algorithms. These functions have desired submodularity and monotonicity properties. However, a natural question is whether minimizing these proxies results in lower MSE. Additionally, which proxies are better suited for the sensor selection in an HSN? We answer these questions in this section.  

Since the CRLB gives a lower bound for MSE, different submodular functions of the CRLB matrix, such as $\text{trace}(\text{CRLB})$, $- \log \det (\text{CRLB})$, and the largest eigenvalue $\lambda_{\text{max}}(\text{CRLB})$ are commonly used as surrogates for MSE \cite{Tohidi_2019_Sparse_antenna}. In the greedy algorithms such as Algorithms~\ref{alg:greedy_alg_gen} and \ref{alg:JGS}, where sensors are sequentially selected starting from a null set, CRLB-based functions may not result in meaningful costs in the first few interactions. For example, consider the linear measurements as in \eqref{eq: problem_model_linear} and let $\mathbf{y}_\mathcal{S}$ be the subset of measurements at a given iteration. Then FIM for estimation of $\mathbf{x}$ from $\mathbf{y}_\mathcal{S}$ is given as $\mathbf{A}_{\mathcal{S}}^{T}\mathbf{A}_{\mathcal{S}}$ where $\mathbf{A}_{\mathcal{S}} \in \mathbb{R}^{|\mathcal{S}|\times K}$ is a matrix consisting of rows of $\mathbf{A}$ that are indexed by the set $\mathcal{S}$. The matrix is singular if $|\mathcal{S}| < K$; hence, CRLB can not be defined. This results in an arbitrary selection of measurements in any greedy algorithm's first $K-1$ iterations. Alternative cost functions such as FP are used to avoid such scenarios, as discussed next.

An FP-based cost depends on the correlation between the measurements \cite{ranieri2014near, Tohidi_2019_Sparse_antenna}, and it does not suffer from the singularity issue. Since correlated measurements add less information about the unknown vector to be estimated, FP-based cost can be minimized for subset selection. Submodular cost functions based on FP are defined for linear and non-linear homogeneous measurements in \cite{ranieri2014near} and \cite{Tohidi_2019_Sparse_antenna}, respectively, where in the latter case, it is assumed that $\mathcal{A}$ is a differentiable function of $\mathbf{x}$. We extend these definitions for heterogeneous measurements and introduce weighted frame potential (WFP) as a performance metric. 

In a greedy selection step, while selecting a new sensor (or measurement) from the remaining ones, any measurement that is highly correlated with the selected sensors or has low SNR should be given less consideration. Following this logic, we define WFP as 
\begin{equation}
    \text{WFP}\left(\mathcal{S}\right) \coloneqq \sum_{i,j \in \mathcal{S}} w_i w_j \frac{\left| \left< \nabla_{\mathbf{x}} y_i, \nabla_{\mathbf{x}} y_j \right> \right|^2}{\| \nabla_{\mathbf{x}} y_i \|_2^2 \| \nabla_{\mathbf{x}} y_j \|_2^2},
    \label{def: WMFP_NL}
\end{equation}
where $\nabla_{\mathbf{x}} y_i$ denotes the gradient of the $i$-th observation with respect to the parameter vector $\mathbf{x}$, and $w_i \coloneqq \phi \left( \sigma_i \right)$ are non-negative weights dependent on the noise variance $\sigma_i$ corrupting the $i$-th measurement. The definition is an extension of the modified frame potential introduced in \cite{Tohidi_2019_Sparse_antenna} and applies to non-linear and linear heterogeneous measurements. In the case of a linear measurement system $\mathbf{A}$, the WFP metric reduces to the simpler form,
\begin{equation}
    \text{WFP}\left(\mathcal{S}\right) \coloneqq \sum_{i,j \in \mathcal{S}} w_i w_j \frac{\left| \left< \mathbf{a}_i, \mathbf{a}_j \right> \right|^2}{\| \mathbf{a}_i \|_2^2 \| \mathbf{a}_j \|_2^2},
    \label{def: WMFP}
\end{equation}
where $\mathbf{a}_i$ denotes the $i$-th row of $\mathbf{A}$. 

The weighting function $\phi \left( \sigma_i \right)$ needs to be chosen such that the measurements corrupted by low noise are preferred for selection. A natural choice is $\phi \left( \sigma_i \right) = 1 / \sigma_i$. This is also the weighting function used in maximum likelihood estimation techniques \cite{kay1993fundamentals}. However, when $\sigma_i \approx 0$, that is when the noise affecting a measurement is negligible, this makes the corresponding terms in WFP extremely large. Since our aim is to minimize the correlation among samples, having $\phi \left( \sigma_i \right) = 1 / \sigma_i$ essentially ensures that the measurements with very low noise are never selected. To address this issue, we choose $\phi \left( \sigma_i \right): \mathbb{R}^+ \rightarrow [0,1]$ to be a function that maps the noise variance to a value within $[0,1]$, such that more weight is given to measurements with low noise. Several such functions have been used, such as $ \frac{1}{1 + \sigma_i}$, $ \frac{1}{1 + \exp \left( - \left( \sigma_i - \Bar{\sigma} \right) \right)}$, and $ \frac{\tanh \left( \sigma_i - \Bar{\sigma} \right)}{2}$, with $\Bar{\sigma}$ representing the mean of the noise variances of all the measurements. It is observed empirically that the sigmoid function, $\frac{1}{1 + \exp \left( - \left( \sigma_i - \Bar{\sigma} \right) \right)}$, gives the best performance among all the $\phi \left( \sigma_i \right)$ used. Hence, we select this as the weighting function for our experiments.    
 
Next, we define an FP-based submodular performance metric that, when maximized, finds an approximate solution for the problem (\ref{eq:relaxed_prob}). Specifically, we define weighted frame cost (WFC) as
\begin{equation}
    \text{WFC}(\mathcal{T}) \coloneqq \text{WFP} (\mathcal{N}) - \text{WFP} \left( \mathcal{N} \backslash \mathcal{T} \right).
    \label{def: WMFC}
\end{equation}
Here, $\mathcal{N} \backslash \mathcal{T}$ denotes the set of measurements that are not in $\mathcal{T}$. In Appendix~\ref{append:a}, we show that WFC is a normalized, monotone, non-decreasing submodular function. In principle, the definition of WFC is similar to the FP-based cost functions used in \cite{ranieri2014near,Tohidi_2019_Sparse_antenna}

The cost, WFC, is used as the performance metric $\mathcal{C}$ in problem (\ref{eq:relaxed_prob}) for the simulations in the following sections. It is to be noted that the maximization of $\text{WFC}(\mathcal{T})$ results in the minimization of $\text{WFP}(\mathcal{N} \backslash \mathcal{T})$ as is apparent from the definition of WFC given by \eqref{def: WMFC}. Since our objective is to minimize the WFP of the selected set, $\mathcal{N} \backslash \mathcal{T}$ is the set of sensors the JGS algorithm selects with WFC as the performance metric.

\subsection{Near Optimality of the WFP}

From Theorem~\ref{thm:Perf_Guar_Hetero_Half}, we know that the WFC of the set $\mathcal{T}^\ast$ of sensors selected by the JGS algorithm is at worst $50\%$ of the optimal WFC. Using this fact, the WFP of the solution can also be bounded, as summarized in the next theorem.

\begin{theorem}\label{thm:WFP_opt_bound}
    If Algorithm~\ref{alg:JGS} selects the set $\mathcal{N} \backslash \mathcal{T}^\ast$ with WFC as the cost function, then
    \begin{equation}
        \text{WFP} \left( \mathcal{T}^\ast \right) \le \frac{1}{2} \left( 1 + \frac{\text{WFP} (\mathcal{N})}{\text{WFP} \left( \mathcal{T}_{OPT} \right)} \right) \text{WFP} \left( \mathcal{T}_{OPT} \right),
        \label{eq:WPF_opt_bound}
    \end{equation}
    where $\mathcal{T}_{OPT}$ is the optimal solution when WFP is minimized.
\end{theorem} 

\begin{proof}
    From Theorem~\ref{thm:Perf_Guar_Hetero_Half}, we have, $\text{WFC} \left( \mathcal{N} \backslash \mathcal{T}^\ast \right) \ge \frac{1}{2} \text{WFC} \left( \mathcal{N} \backslash \mathcal{T}_{OPT} \right)$. This implies,
    \begin{align}
        & \text{WFP} (\mathcal{N}) - \text{WFP} \left( \mathcal{T}^\ast \right) \ge \frac{1}{2} \left( \text{WFP} (\mathcal{N}) - \text{WFP} \left( \mathcal{T}_{OPT} \right) \right) \notag \\
        \implies & \text{WFP} \left( \mathcal{T}^\ast \right) \le \frac{1}{2} \left( \text{WFP} (\mathcal{N}) + \text{WFP} \left( \mathcal{T}_{OPT} \right) \right), \notag
    \end{align}
    and the result follows.
\end{proof}

As is evident from \eqref{eq:WPF_opt_bound}, the WFP of $\mathcal{T}^\ast$ is at most a factor $\gamma \coloneqq \frac{1}{2} \left( 1 + \frac{\text{WFP}(\mathcal{N})}{\text{WFP}(\mathcal{T}_{OPT})} \right)$ of the optimal WFP. Since $\text{WFP} (\mathcal{N}) \ge \text{WFP}(\mathcal{T}_{OPT})$, this implies $\gamma \ge 1$. The value of $\gamma$ depends on the value of $\text{WFP} (\mathcal{N}) / \text{WFP}(\mathcal{T}_{OPT})$. In the linear case, when the measurements are obtained using a measurement matrix $\mathbf{A}$, $\text{WFP}(\mathcal{N})$ can be bounded as \cite{ranieri2014near},
\begin{equation}
    \frac{1}{K}\left( \sum_{i=1}^N w_i^2 \| \mathbf{a}_i \|^2 \right)^2 \le \text{WFP} (\mathcal{N}) \le \left( \sum_{i=1}^N w_i^2 \| \mathbf{a}_i \|^2 \right)^2,\nonumber
    \label{eq:WFP_entire_mat_bound}
\end{equation}
where the lower bound is reached when $\Tilde{\mathbf{A}} \coloneqq \text{diag}\left( \{w_i\}_{i=1}^N \right) \mathbf{A}$ is a tight frame, and the upper bound is reached when $\mathbf{A}$ is a rank-one matrix.

Thus, the JGS algorithm not only achieves near-optimal WFC but also reaches near-optimal WFP, which we want to minimize. However, since MSE is the performance criterion used in most estimation tasks, we next prove a theoretical bound on the MSE achieved by the JGS algorithm when the measurements follow a linear model.

\subsection{Near Optimality of the MSE with respect to WFP}

As mentioned earlier, mean squared error is a widely used error measure to compute the performance of estimation methods. In this section, We prove theoretical bounds on the MSE achieved by the JGS algorithm when WFC is used as the performance metric and the measurement model is linear.

In order to compute the approximation factor for the MSE achieved by the JGS algorithm, we first introduce the MSE metric and the estimation method used. For an estimate $\hat{\mathbf{x}}$ of the parameter vector $\mathbf{x}$, the MSE can be defined as $\text{MSE} \coloneqq \| \mathbf{x} - \hat{\mathbf{x}} \|_2^2$. In the case when the sensing system is linear, so that $\mathcal{A}$ is a matrix $\mathbf{A}$, for a heterogeneous system with known noise covariance matrix $\boldsymbol{\Sigma}$ and a subset $\mathcal{T}$ of selected $M = \sum_{i=1}^L M_i$ measurements, the minimum MSE is achieved by the estimate
\begin{equation}
    \hat{\mathbf{x}} \coloneqq \left( \mathbf{A}_{\mathcal{T}}^H \boldsymbol{\Sigma}_{\mathcal{T}}^{-1} \mathbf{A}_{\mathcal{T}} \right)^{-1} \mathbf{A}_{\mathcal{T}}^H \boldsymbol{\Sigma}_{\mathcal{T}}^{-1} \mathbf{y}_{\mathcal{T}},
    \label{eq:MMSE_est_hetero}
\end{equation}
and the corresponding MSE is given by $\text{MSE} (\mathcal{T}) = \text{Trace} \left( \left( \mathbf{A}_{\mathcal{T}}^H \boldsymbol{\Sigma}_{\mathcal{T}}^{-1} \mathbf{A}_{\mathcal{T}} \right)^{-1} \right)$ \cite{golub1963comparison}. Here, $\mathbf{A}_{\mathcal{T}}$ (or $\mathbf{y}_{\mathcal{T}}$) denotes the rows of $\mathbf{A}$ (or $\mathbf{y}$) indexed by the set $\mathcal{T}$. The corresponding noise covariance matrix is defined as  $\boldsymbol{\Sigma}_{\mathcal{T}} \coloneqq \text{diag} \left( \left\{ \sigma_i^2 | i \in \mathcal{T} \right\} \right)$. Let us define the inverse error covariance matrix of the estimate $\hat{\mathbf{x}}$ as $\boldsymbol{\Phi}_{\mathcal{T}} \coloneqq \left( \mathbf{A}_{\mathcal{T}}^H \boldsymbol{\Sigma}_{\mathcal{T}}^{-1} \mathbf{A}_{\mathcal{T}} \right)$. Then minimizing $\text{MSE} (\mathcal{T})$ is equivalent to minimizing $\text{Trace}(\boldsymbol{\Phi}_{\mathcal{T}}^{-1})$, or equivalently minimizing $\sum_{i=1}^M 1/\lambda_i^{\mathcal{T}}$, where $\lambda_i^{\mathcal{T}}$, $i = 1, \ldots, M$, are the eigenvalues of $\boldsymbol{\Phi}_{\mathcal{T}}$. 

We derive the bounds for the MSE of the estimate only for a certain class of measurement matrices with a bounded spectrum of their corresponding $\boldsymbol{\Phi}_{\mathcal{T}}$ matrix. This class of measurement matrices is called $(\delta, M)$-bounded frames.

\begin{definition}{($(\delta,M)$-Bounded Frame) \cite{ranieri2014near}}
    A matrix $\mathbf{A} \in \mathbb{C}^{N \times K}$, with $N \ge M > K$, is said to be $(\delta, M)$-bounded if, for every $\mathcal{T} \subseteq \mathcal{N}$ such that $| \mathcal{T} | = M$, and corresponding noise covariance matrix $\boldsymbol{\Sigma}_{\mathcal{T}}$, the corresponding matrix $\boldsymbol{\Phi}_{\mathcal{T}}$ has a bounded spectrum
    \begin{equation}
        d - \delta \le \lambda_i^{\mathcal{T}} \le d + \delta,
    \end{equation}
    for all $1 \le i \le M$, $\delta > 0$, and where $d = \frac{1}{N} \sum_{i \in \mathcal{N}} \| \mathbf{a}_i \|^2$ is the average squared norm of the rows of $\mathbf{A}$.
\end{definition}

A $(\delta, M)$-bounded frame is reminiscent of the restricted isometry property in compressed sensing \cite{eldar_cs_book}. For a $(\delta, M)$-bounded class of measurement matrices, we can state the following bounds for the MSE of the estimate.

\begin{theorem} \label{thm:mse_opt_bound}
    If $\mathbf{A}$ is a $(\delta, M)$-bounded frame, and Algorithm~\ref{alg:JGS} selects the set $\mathcal{N} \backslash \mathcal{T}^\ast$ with WFC as the cost function, then
    \begin{equation}
        \text{MSE} \left( \mathcal{T}^\ast \right) \le \zeta \text{MSE}_{OPT},
        \label{eq:mse_opt_bound}
    \end{equation}
    where $\text{MSE}_{OPT}$ is the optimal MSE that can be achieved, $\zeta = \gamma \kappa$ where $\gamma$ is the approximation factor in \eqref{eq:WPF_opt_bound} and $\kappa = \frac{(d+\delta)^2}{(d-\delta)^2} \frac{\alpha_{\max}}{\alpha_{\min}}$, with $\alpha_{\max} = \underset{\mathcal{T} \subset \mathcal{N}, |\mathcal{T}| = M}{\max} \quad \sum_{i \in \mathcal{T}} w_i^2 \| \mathbf{a}_i \|^2$, and $\alpha_{\min} = \underset{\mathcal{T} \subset \mathcal{N}, |\mathcal{T}| = M}{\min} \quad \sum_{i \in \mathcal{T}} w_i^2 \| \mathbf{a}_i \|^2$.
\end{theorem} 

\begin{proof}
    We know that, $\text{WFP} (\mathcal{T}) = \text{Trace} \left( \boldsymbol{\Phi}_{\mathcal{T}}^H \mathbf{W}_{\mathcal{T}}^H \mathbf{W}_{\mathcal{T}} \boldsymbol{\Phi}_{\mathcal{T}} \right) = \sum_{j=1}^M \left| w_j \lambda_j^{\mathcal{T}} \right|^2$, where $W_{\mathcal{T}} = \text{diag} \left( \left\{ \sigma_i w_i | i \in \mathcal{T} \right\} \right)$ and $w_j$ is a linear combination of the elements in $\left\{ \sigma_i w_i | i \in \mathcal{T} \right\}$. Using the relationship \cite{sharma2008some} between the arithmetic mean, harmonic mean, and the standard deviation of the $\left\{ \lambda_i^{\mathcal{T}} \right\}_{i=1}^M$, and the fact that $w_i \le 1$, for $1 \le i \le N$, we get,
    \begin{equation}
        \frac{M^2}{\alpha_{\max}} \frac{\text{WFP} (\mathcal{T})}{\lambda_{\max}^2} \le \text{MSE} (\mathcal{T}) \le \frac{M^2}{\alpha_{\min}} \frac{\text{WFP} (\mathcal{T})}{\lambda_{\min}^2},
        \label{eq:MSE_am_hm_bound}
    \end{equation}
    where $\lambda_{\max}$ and $\lambda_{\min}$ are the maximum and minimum eigenvalues of $\boldsymbol{\Phi}_{\mathcal{T}}$ respectively. This implies that,
    \begin{equation}
        \text{MSE} \left( \mathcal{T}^\ast \right) \le \frac{M^2}{\alpha_{\min}} \frac{\gamma \text{WFP} \left( \mathcal{T}_{OPT} \right)}{(d-\delta)^2},
        \label{eq:MSE_upp_bound}
    \end{equation}
    and,
    \begin{equation}
        \text{MSE}_{OPT} \ge \frac{M^2}{\alpha_{\max}} \frac{\text{WFP} \left( \mathcal{T}_{OPT} \right)}{(d+\delta)^2},
        \label{eq:MSE_low_bound}
    \end{equation}
    Combining \eqref{eq:MSE_upp_bound} and \eqref{eq:MSE_low_bound}, we get \eqref{eq:mse_opt_bound}.
\end{proof}

The proof of Theorem \ref{thm:mse_opt_bound} follows an approach similar to the one used in \cite[Theorem~3]{ranieri2014near} where $\text{FP}(\mathcal{T})$ is used instead of $\text{WFP}(\mathcal{T})$. Note that $\text{FP}(\mathcal{T})$ is a function of $\boldsymbol{\Phi}_{\mathcal{T}}$ and hence its eigenvalues. On the other hand, $\text{WFP}(\mathcal{T})$ is a function of $\mathbf{W}_{\mathcal{T}} \boldsymbol{\Phi}_{\mathcal{T}}$, which is not a function of the eigenvalues of $\boldsymbol{\Phi}_{\mathcal{T}}$. Thus, the major challenge in deriving \eqref{eq:MSE_am_hm_bound}. 

A bound on the MSE has also been obtained in homogeneous sensor selection by exploiting its weak submodularity property and using a randomized greedy algorithm \cite{hashemi_2020_randomized}. This bound depends on the curvature \cite[Definition~3]{hashemi_2020_randomized} of the MSE. Such results can be extended to the heterogeneous setting but are not explored in this work due to space constraints. Instead, from Theorem~\ref{thm:mse_opt_bound}, we can conclude that the MSE of the solution is bounded compared to the optimal MSE since the WFP achieved is bounded with respect to the optimal WFP, that is, $\gamma$ is bounded. 

In the next section, we experimentally validate our results for both linear and non-linear measurement models.

\section{Experiments and Results} \label{sec:Experiments}
This section presents a detailed analysis of the proposed algorithm through simulations for linear and nonlinear measurement models. The results are compared with different methods such as greedy selection, random selection, independent greedy selection, and independent random selection (cf. Table \ref{tab: MSE_measures} for the list of methods and see Section~\ref{sec:JGS} for details). We apply an exhaustive search to find an optimal solution for a small-scale problem. The methods are compared in terms of normalized mean-squared error measured as
\begin{align}
    \text{MSE} = 10 \log_{10} \left(\frac{\|\mathbf{x} - \mathbf{\hat{x}}\|^2 }{\|\mathbf{x}\|^2} \right),
    \label{eq:mse}
\end{align}
where $\mathbf{\hat{x}}$ is an estimate of $\mathbf{x}$ given by \eqref{eq:MMSE_est_hetero}.

\begin{table}
    \caption{Different methods used for comparing the results}
    \label{tab: MSE_measures}
    \begin{center}
    \adjustbox{max width=\linewidth}{%
    \begin{tabular}{|l|l|}
        \hline
        \texttt{OPT} & Optimal Solution (Exhaustive Search) \\
        \hline
        \texttt{JGS} & Joint Greedy Selection Algorithm \\
        \hline
        \texttt{GS} & Greedy Selection Algorithm \\
        \hline
        \texttt{IGS} & Independent Greedy Selection Algorithm \\
        \hline
        \texttt{IRS} & Independent Random Selection Algorithm\\
        \hline
        \texttt{RS} & Random Selection Algorithm \\
        \hline
        
    \end{tabular}
    }
    
    \end{center}
\end{table}

\subsection{Linear Measurement Model}

We first present simulation results when the measurement model is linear in nature as in \eqref{eq: problem_model_linear}. In the experiments, $\mathbf{A}$ is a subsampled discrete cosine transform (DCT) matrix obtained by randomly selecting $K$ columns from a $N \times N$ DCT matrix. The entries of the parameter vector $\mathbf{x}\in\mathbb{R}^K$ are independent and identically distributed (i.i.d.) zero-mean Gaussian random variables with variance 25. For the additive noise model in \eqref{eq: problem_model_linear}, the $L$ noise variances are specified in terms of the SNRs as
\begin{equation}
    \text{SNR}_i \coloneqq \frac{1}{N_i \sigma_i^2} \sum_{j \in \mathcal{S}_i} \left( [\mathcal{A}(\mathbf{x})]_j \right) ^2, \quad \text{for } i =1, \ldots, L, 
    \label{eq:SNR_i}
\end{equation}
where $[\mathcal{A}(\mathbf{x})]_j$ is the noise-free measurement at the $j$-th sensor. 
MSEs are averaged over $1000$ Monte-Carlo iterations for each set of SNRs. For each Monte Carlo iteration, $\mathbf{x}$, $\mathbf{A}$, and $\text{SNR}_i$ are kept fixed, while the noise covariance matrix $\boldsymbol{\Sigma}$ is changed by randomly choosing the locations of the $N_i$ sensors in each set $\mathcal{S}_i$ out of the total $N$ possible locations.

\begin{figure}
    \centering
    \includegraphics[width= 2.6 in]{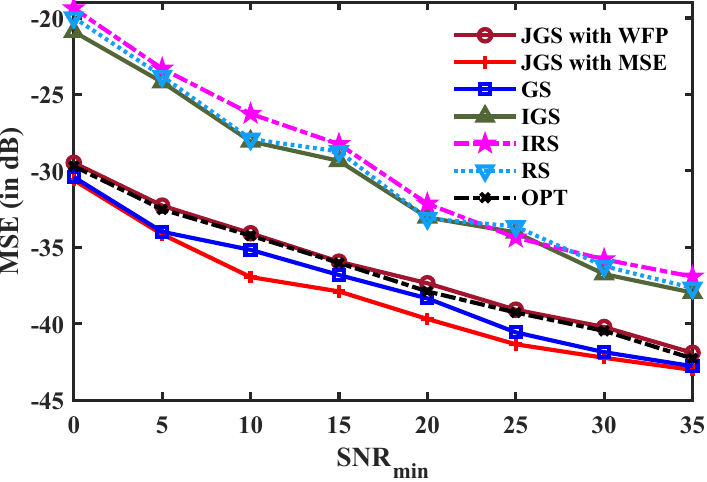}
    \caption{
    Small-scale problem with linear measurements in additive noise: Comparison of MSEs achieved by the different algorithms. Parameters: $L = 3$, $M_1 = 3$, $M_2 = 5$, $M_3 = 2$, $N_1 = 5$, $N_2 = 10$, $N_3 = 5$, $K = 5$, $\text{SNR}_{1} = 40$ dB, $\text{SNR}_{2} = \text{SNR}_{\min}$, and $\text{SNR}_{3} = \left( \text{SNR}_{1} + \text{SNR}_{2} \right) / 2$. The JGS has lower MSE than IGS, IRS, and RS by $5$-$8$ dB, and higher MSE than GS and OPT by $0.2$-$2$ dB.}
    \label{fig:Perf_Comp_N_15_LF}
    \vspace{-0.1in}
\end{figure}

For a small-scale problem, the performances of the different algorithms mentioned in Table~\ref{tab: MSE_measures} are compared with the optimal solution $\mathcal{T}_{OPT}$ obtained through an exhaustive search. In this problem, the goal is to select $M_1 = 3$ sensors from $N_1 = 5$ sensors, $M_2 = 5$ out of $N_2 = 10$ sensors, and $M_3 = 2$ sensors from $N_3 = 5$ sensors for $K = 5$. We fix $\text{SNR}_1 = 40$ dB and vary $\text{SNR}_2$ from $0$ dB to $35$ dB. $\text{SNR}_3$ is chosen as the average of $\text{SNR}_1$ and $\text{SNR}_2$. The MSEs of different methods are shown in Fig.~\ref{fig:Perf_Comp_N_15_LF}. We observe that compared to IGS and the random selection methods (RS and IRS), GS and the proposed JGS method have $5$-$8$ dB lower error. Among the performance metrics used, MSE as an objective function gives $0.5-2.5$ dB lower MSE as compared to the WFC metric. This shows that JGS performs equally well with non-submodular metrics, although no theoretical guarantees are provided for the same. GS results in $0.5-2$ dB lower error among GS and JGS methods. This is intuitively reasonable since the heterogeneous sensor selection problem is more constrained than the homogeneous sensor selection problem. Specifically, by applying the GS algorithm, we can select $M$ measurements from $N$ measurements but can not ensure that out of the selected sensors, $M_1$, $M_2$ and $M_3$ are selected from the sets $\mathcal{S}_1$, $\mathcal{S}_2$ and $\mathcal{S}_3$ respectively. As the difference between $\text{SNR}_1$ and $\text{SNR}_2$ decreases, the difference between the performances of GS and JGS reduces as the sets become more homogeneous in nature. The optimal solution gives an MSE about $0.2-0.5$ dB less than that JGS achieves. Thus, the proposed algorithm achieves near-optimal performance.

\begin{figure}[tb]
    \centering
        \includegraphics[width=3.2 in]{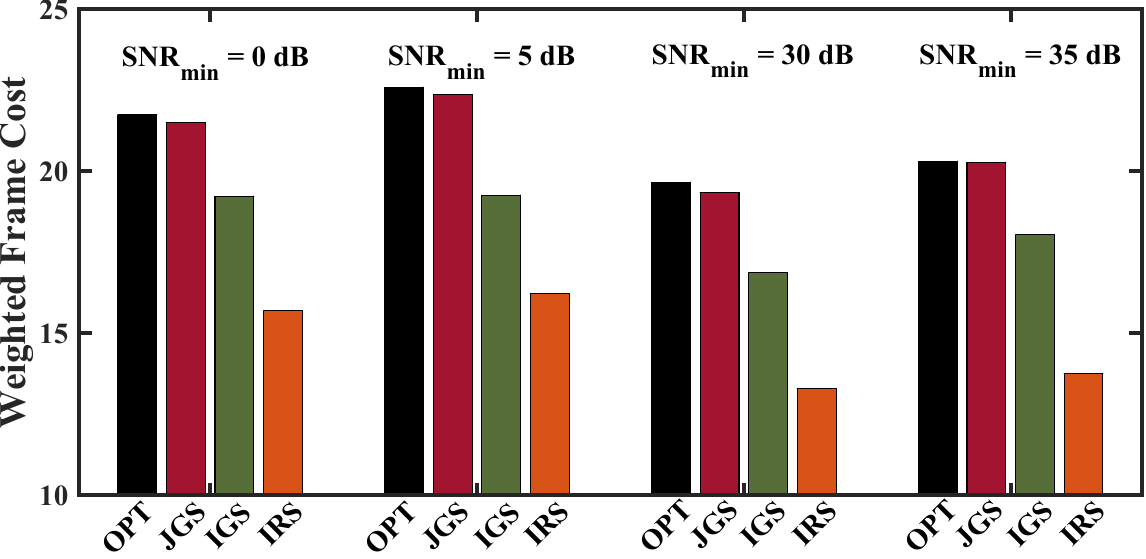}
    \caption{Comparison of WFCs of JGS, IGS, and IRS with respect to the optimal performance when $\text{SNR}_2 = \text{SNR}_{\min} = 0, 5, 30, 35$ dB. Parameters: $M_1 = 3$, $M_2 = 5$, $M_3 = 2$, $N_1 = 5$, $N_2 = 10$, $N_3 = 5$, and $K = 5$. $\text{SNR}_1 = 40$ dB and $\text{SNR}_3 = \left( \text{SNR}_1 + \text{SNR}_2 \right) / 2$. JGS gives near-optimal WFC, while the WFC of IGS and IRS are comparatively lower.}
    \label{fig:Alg_Comp_w_OPT_LF}
    \vspace{-0.1 in}
\end{figure}

While simulating the small-scale problem, we also take the opportunity to verify the weighted frame cost achieved by JGS, IGS, and IRS as compared to the optimal WFC. Figure~\ref{fig:Alg_Comp_w_OPT_LF} shows the WFCs for $\text{SNR}_1 = 40$ dB and $\text{SNR}_2 = 0$ dB, $5$ dB, $30$ dB, and $35$ dB. We observe that the solution sets obtained by the JGS algorithm achieve performances equivalent to the optimal. In comparison, IGS and IRS produce lower WFCs. This experiment shows that, in practice, the proposed JGS algorithm achieves near-optimal performance. Figure~\ref{fig:Perf_Comp_N_15_LF} and Fig.~\ref{fig:Alg_Comp_w_OPT_LF} show that MSE and WFC are inversely correlated with each other, in the sense that they follow inverse trends. A lower value of WFC indicates a higher MSE.

\begin{figure}
    \centering
    \includegraphics[width= 2.6 in]{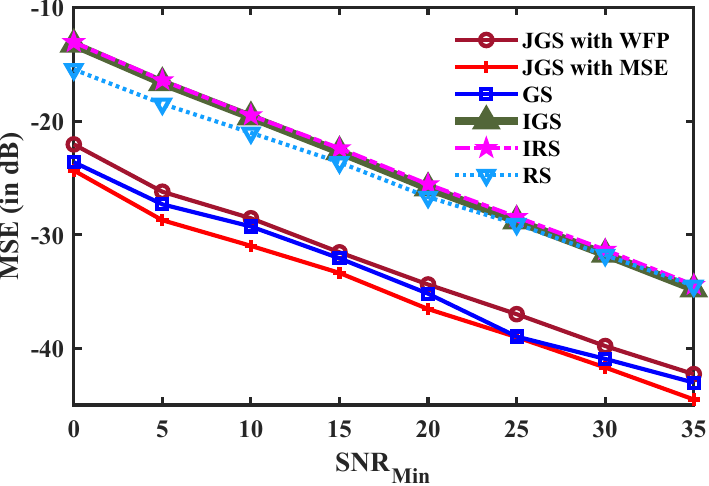}
    \caption{Large-scale problem with linear measurements in additive noise: Comparison of MSEs achieved by the different algorithms. Parameters: $M_i = \{10, 10, 10, 60, 10\}$, $N_i = \{ 25, 25, 25, 100, 25\}$, $\text{SNR}_1 = 40$ dB, $\text{SNR}_5 = \text{SNR}_{\text{Min}}$, and $\text{SNR}_2$, $\text{SNR}_3$ and $\text{SNR}_4$ are equally spaced between $\text{SNR}_1$ and $\text{SNR}_5$. JGS has lower MSE than IGS, IRS, and RS by $4$-$8$ dB, and higher MSE than GS by $0.5$-$1.5$ dB.}
    \label{fig:Perf_Comp_N_200_LF}
\end{figure}

In the next experiment, we focused on a large-scale problem with $M_i = \{ 10, 10, 10, 60, 10 \}$, $N_i = \{ 25, 25, 25, 100, 25 \}$, for $i = 1, \ldots, 5$, and $K = 30$. Due to the size of the problem, an exhaustive search could not be applied. $\text{SNR}_1$ is kept fixed at $40$ dB, while $\text{SNR}_5$ is varied from $0-35$ dB. $\text{SNR}_2$, $\text{SNR}_3$ and $\text{SNR}_4$ are chosen to be uniformly spaced between $\text{SNR}_1$ and $\text{SNR}_5$. The MSEs for the rest of the methods are compared in Fig.~\ref{fig:Perf_Comp_N_200_LF}. We observe that the JGS approach gives $4$-$8$ dB lower error than the IGS, RS, and IRS methods for different noise levels.

\begin{figure}[tb]
    \centering
    \includegraphics[width=2.6 in]{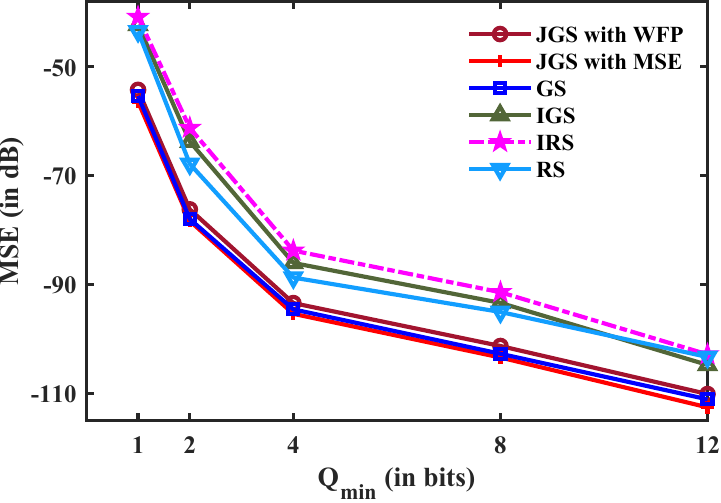}
    \caption{Large-scale problem with quantized linear measurements: Comparison of MSEs achieved by the different algorithms. Parameters: $M_i = \{ 10, 10, 10, 60, 10 \}$, $N_i = \{ 25, 25, 25, 100, 25 \}$, for $i = 1, \ldots, 5$, $K = 30$, $Q_1 = 16$ bits, and $Q_5 = Q_{\min}$. $Q_2$, $Q_3$ and $Q_4$ are chosen to lie between $Q_1$ and $Q_5$. JGS has lower MSE than IGS, IRS, and RS by $3-8$ dB, and higher MSE than GS by $0.5-2.5$ dB.}
    \label{fig:Perf_Comp_N_200_LF_Quant}
    \vspace{-0.1in}
\end{figure}

Next, we consider the large-scale problem as earlier, but instead of adding two levels of noise, we consider quantizing the measurements with a coarse and a fine quantizer. The simulation settings consider the practical scenario where high-precision and low-precision sensors may also be characterized by the quantization levels of their analog-to-digital converters (ADCs). Although quantization error is often modeled as random additive noise, the operation is non-linear. In this experiment, the sensors in the set $\mathcal{S}_1$ are quantized with $Q_1 = 16$ bits, while the sensors in set $\mathcal{S}_5$ are quantized with $Q_5 = Q_{\min}$ bits, where $Q_{\min}$ varies from $1-12$ bits. $Q_2$, $Q_3$ and $Q_4$ are chosen to be uniformly spaced integers within the range $\left( Q_{\min}, 16 \right)$. Figure~\ref{fig:Perf_Comp_N_200_LF_Quant} shows the results of this experiment of estimating the parameter vector $\mathbf{x}$ from its quantized linear measurements. As in the previous cases, JGS gives lower MSE than random and independent greedy methods by $3-8$ dB.

\begin{figure}[tb]
    \captionsetup[subfigure]{font={scriptsize,rm},labelfont={footnotesize,rm}}
    \centering
    \begin{subfloat}[Small-scale problem]{
        \includegraphics[width= 2.6 in]{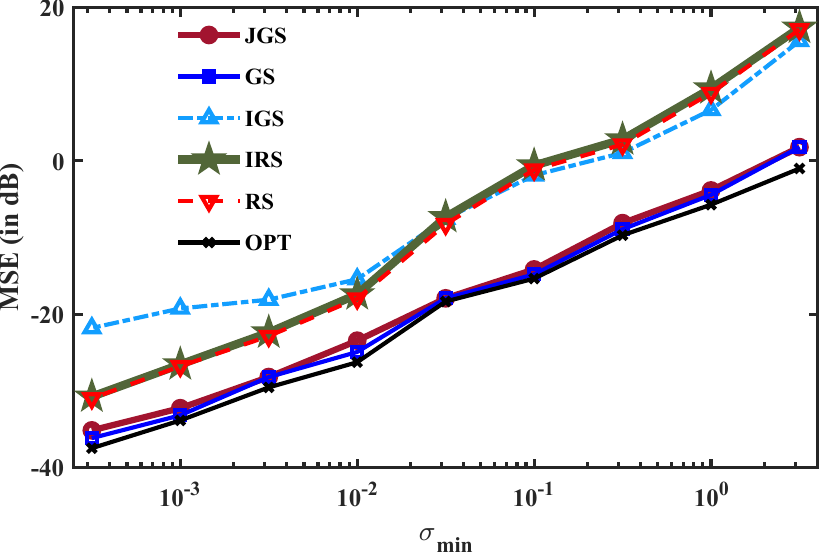}%
        \label{fig:Perf_Comp_N_15_LF_additive_noise}%
    }
    \end{subfloat}
    \vspace{0.1in}
    \begin{subfloat}[Large-scale problem]{
        \includegraphics[width= 2.6 in]{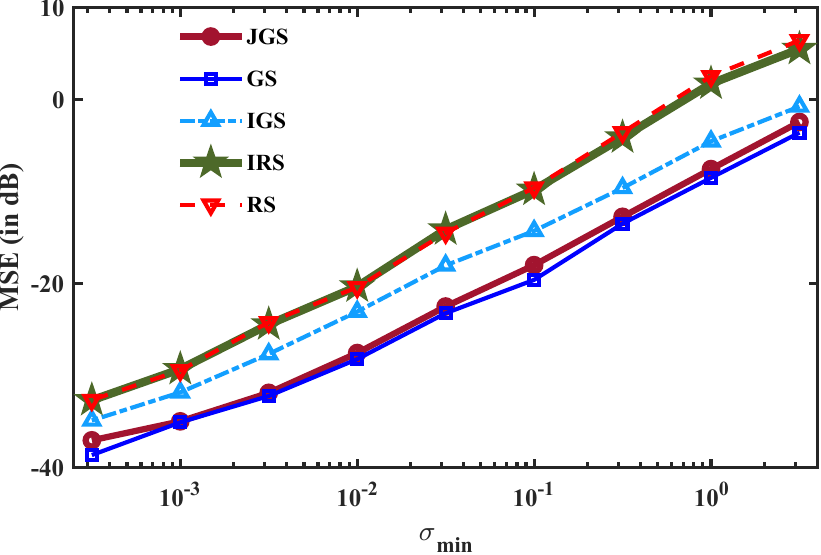}%
        \label{fig:Perf_Comp_N_200_LF_additive_noise}%
    }
    \end{subfloat}
    \caption{Comparison of MSEs of different algorithms for varying $\mathbf{x}$: (a) Parameters: $L = 3$, $M_i = \{3, 5, 2\}$, $N_i = \{5, 10,5\}$, for $i = 1,2,3$, $K = 5$, $\sigma_{_1} = 10^{-4}$, and $\sigma_{_2} = \sigma_{\min}$. The noise variance of $\mathcal{S}_3$ is uniformly spaced between $\sigma_{_1}$ and $\sigma_{_2}$. JGS has lower MSE than IGS, IRS, and RS by $8$-$9$ dB and higher MSE than GS and OPT by $0.5$-$1.5$ dB. (b) Parameters: $L=5$, $M_i = \{10, 10, 10, 60, 10\}$, $N_i = \{25, 25, 25, 100, 25\}$, for $i = 1, \ldots, 5$, $K = 30$, $\sigma_{_1} = 10^{-4}$, and $\sigma_{_5} = \sigma_{\min}$. The noise variances of $\mathcal{S}_2$, $\mathcal{S}_3$ and $\mathcal{S}_4$ are uniformly spaced between $\sigma_{_1}$ and $\sigma_{_5}$. JGS has lower MSE than IGS, IRS, and RS by $4$-$9$ dB, and higher MSE than GS by $1$-$1.5$ dB.}
    \label{fig:Perf_Comp_Varying_x}

    \vspace{-0.1in}
\end{figure}

In the previous experiments, the results are obtained for a fixed $\mathbf{x}$, while the noise covariance matrix $\boldsymbol{\Sigma}$ and the measurement noise at each Monte Carlo trial is changed. In the following experiment, we compare the algorithms for different $\mathbf{x}$ realizations where $\boldsymbol{\Sigma}$ is fixed. In this case, the measurements are corrupted by additive zero-mean Gaussian noise. The noise variances for $\mathcal{S}_1$ and $\mathcal{S}_5$ are chosen as $\sigma_1 = 10^{-4}$ and $\sigma_5 = \sigma_{\min}$, with $\sigma_{\min}$ varying from $3 \times 10^{-4}$ to $3$. The noise variances of $\mathcal{S}_2$, $\mathcal{S}_3$ and $\mathcal{S}_4$ are chosen to be uniformly spaced between $\sigma_1$ and $\sigma_5$. The average MSE is calculated over $1000$ Monte Carlo trials where in each trial, entries of $\mathbf{x}$ are chosen from an i.i.d. uniform distribution selected from the range $[-1,1]$. This ensures that the noise-free measurements are bounded, thus making the SNR values of the two sets finite. Comparison of average MSEs for different methods are presented in Fig.~\ref{fig:Perf_Comp_Varying_x} for both small- and large-scale problems with the values of $M_i$, $N_i$, for $i=1, \ldots, L$, and $K$ as used in the previous experiments. It is observed that JGS performs better than IGS, RS, and IRS by $4$-$9$ dB, while GS outperforms JGS by about $0.5$-$1.5$ dB. These results show that the set of sensors chosen by the proposed JGS algorithm can be used to estimate the underlying parameter vector $\mathbf{x}$ coming from a given distribution as long as the measurement matrix $\mathbf{A}$ and the noise characteristics given by $\boldsymbol{\Sigma}$ remain the same.

\subsection{Non-linear Measurement Model}
To compare the methods for non-linear measurement models, we consider the direction of arrival (DoA) estimation problem. In this model, we assume that there are $K$ far-field sources transmitting from $K$ different directions $\boldsymbol{\theta} = [\theta_1, \theta_2, \ldots, \theta_K] \in \mathbb{R}^K$. The receiving antenna array is composed of two different sets of sensors with different noise characteristics. The problem is then to select the best set of $M_1$ high-precision sensors and $M_2$ low-precision sensors from the linear sensor array so that the MSE of the estimated DoA is minimized. The measurements are related to the directions as  
\begin{equation}
    \mathbf{y} = \mathbf{A} \left( \boldsymbol{\theta} \right) \boldsymbol{\alpha} + \boldsymbol{\eta} \in \mathbb{C}^N,
\end{equation}
where $\mathbf{A} \left( \boldsymbol{\theta} \right)$ is the array sensing matrix whose $(nk)$-th entry is $e^{-j2\frac{\pi}{\lambda}d_n \sin \left( \theta_k \right)}$, for $k = 1,\ldots, K$, $n = 1, \ldots, N$, where $d_n$ denotes the location of the sensors measured from a reference sensor taken at the origin, and $\lambda$ is the wavelength of the transmitted waves.
The vector $\boldsymbol{\alpha} = [\alpha_1, \ldots, \alpha_K]^{T} \in \mathbb{R}^K$ consists of the amplitudes of the waveform received from the $K$ sources. The vector $\boldsymbol{\eta}$ represents the additive noise in the measurements. The goal is to select $M_1$ high-precision sensors and $M_2$ low-precision sensors from the given sensor array, such that the error in the estimation of $\boldsymbol{\theta}$ and $\boldsymbol{\alpha}$ is minimized. 

\begin{figure}
    \centering
    \includegraphics[width=2.6 in]{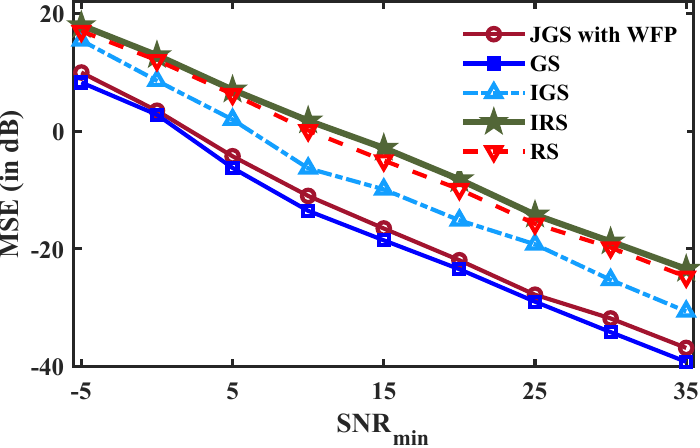}
    \caption{Large-scale problem with non-linear measurements (DoA estimation): Comparison of MSEs achieved by the different algorithms. Parameters: $L = 5$, $M_i = \{10, 10, 10, 60, 10\}$, $N_i = \{25, 25, 25, 100, 25\}$, for $i = 1, \ldots, 5$, $K = 30$, $\text{SNR}_1 = 40$ dB, $\text{SNR}_5 = \text{SNR}_{\min}$, and $\text{SNR}_2$, $\text{SNR}_3$ and $\text{SNR}_4$ are equally spaced between $\text{SNR}_1$ and $\text{SNR}_5$. JGS has lower MSE than IGS, IRS, and RS by $5$-$15$ dB, and higher MSE than GS by $0.5$-$1$ dB.}
    \label{fig:Perf_Comp_DoA}
\end{figure}

In these experiments, simulations are performed to select $M_i = \{10, 10, 10, 60, 10\}$ sensors out of $N_i = \{25, 25, 25, 100, 25\}$ sensors, for $i = 1, \ldots, 5$, for estimating the directions and amplitudes of $15$ sources by using MUSIC algorithm \cite{bienvenu1983optimality, tuncer2009classical}. Thus, the number of parameters to estimate is $K = 30$. 
The simulations are performed with the sources transmitting at $77$-GHz, which corresponds to $\lambda = 4$ mm. The sensor locations $d_n$, $n = 1, \ldots, N$ are chosen uniformly at random within the interval $[0,1]$. In this experiment, we set $\text{SNR}_1 = 40$ dB, and $\text{SNR}_5$ being varied. $\text{SNR}_2$, $\text{SNR}_3$ and $\text{SNR}_4$ are uniformly spaced between $\text{SNR}_1$ and $\text{SNR}_5$. The entries of $\boldsymbol{\theta}$ are chosen uniformly at random from the interval $[-\pi, \pi]$, and entries of $\boldsymbol{\alpha}$ are i.i.d. Gaussian random variables with zero mean and variance 25. Figure~\ref{fig:Perf_Comp_DoA} shows the results of this experiment. As in the linear cases, the algorithm is shown to perform almost as well as GS ($0.5$-$1$ dB worse) even when the measurement model is non-linear and vastly outperforms the random selection methods and the IGS method (by $5$-$15$ dB).

In conclusion, the above experiments show that the proposed JGS algorithm gives a near-optimal performance; the MSE it achieves is much lower than random solutions and independent greedy solutions for both linear and non-linear measurement models. Being a less-constrained problem, GS gives a slightly lower MSE than JGS.

\section{Conclusion} \label{sec:Conclusion}

In this work, we have considered the problem of sensor selection in heterogeneous sensor networks when the number of sensors to choose from each subset is specified as a cardinality constraint. We proposed a joint greedy algorithm to select the sensors from the different subnetworks. We proved theoretically that the worst-case performance of the JGS algorithm is $50\%$ of the optimal performance. In the special case when the HSN comprises two subnetworks, asymptotically, it is seen that the worst-case error reaches $(1-1/e) \approx 63\%$ of the optimal performance when $M_1 / M_2 \ll 1$.

We also showed using simulations that the algorithm works well in practice and performs better than the predictions of the worst-case error bound. The small-scale experiments show that JGS gives a near-optimal solution. Both the small-scale and large-scale network simulations show that JGS gives lower MSE than random selection methods and IGS by $4$-$10$ dB in linear and non-linear measurement systems. The proposed algorithm thus can select sensors well in heterogeneous sensing environments.

\appendices

{

\renewcommand\qedsymbol{$\blacksquare$}
\section{Submodularity of WFC}
\label{append:a}

    The function WFC as given by (\ref{def: WMFC}) is normalized since $\text{WFC}(\varnothing) = \text{WFP}(\mathcal{N}) - \text{WFP}(\mathcal{N}) = 0$.
    
    Let $\mathcal{S}_1 \subset \mathcal{S}_2 \subset \mathcal{N}$. Then,
    \begin{align}
        \text{WFC}(\mathcal{S}_2) &- \text{WFC}(\mathcal{S}_1) = \text{WFP}(\mathcal{N} \backslash \mathcal{S}_1) - \text{WFP}(\mathcal{N} \backslash \mathcal{S}_2) \\
        &= \sum_{i \in \mathcal{S}_2 \backslash \mathcal{S}_1} \sum_{j \in \mathcal{N} \backslash \mathcal{S}_1} w_i w_j \frac{\left| \left< \nabla_{\mathbf{x}} y_i, \nabla_{\mathbf{x}} y_j \right> \right|^2}{\| \nabla_{\mathbf{x}} y_i \|_2^2 \| \nabla_{\mathbf{x}} y_j \|_2^2} \ge 0,
    \end{align}
    since $w_i \ge 0$ for all $i \in \mathcal{N}$. This shows that WFC is a monotone, non-decreasing function.
    
    Finally, with $\mathcal{S}_1 \subset \mathcal{S}_2 \subset \mathcal{N}$ and $\rho_j(\mathcal{S})$ defining the incremental cost of adding the element $j$ to the set $\mathcal{S}$ with WFC as the cost function, the submodularity of WFC is proved.
    \begin{align*}
        \rho_j &(\mathcal{S}_1) - \rho_j (\mathcal{S}_2) \\
        &= \sum_{i \in \mathcal{N} \backslash \mathcal{S}_1} w_i w_j \frac{\left| \left< \nabla_{\mathbf{x}} y_i, \nabla_{\mathbf{x}} y_j \right> \right|^2}{\| \nabla_{\mathbf{x}} y_i \|_2^2 \| \nabla_{\mathbf{x}} y_j \|_2^2} \\
        &\hspace{2cm} - \sum_{i \in \mathcal{N} \backslash \mathcal{S}_2} w_i w_j \frac{\left| \left< \nabla_{\mathbf{x}} y_i, \nabla_{\mathbf{x}} y_j \right> \right|^2}{\| \nabla_{\mathbf{x}} y_i \|_2^2 \| \nabla_{\mathbf{x}} y_j \|_2^2} \\
        &= \sum_{i \in \mathcal{S}_2 \backslash \mathcal{S}_1}  w_i w_j \frac{\left| \left< \nabla_{\mathbf{x}} y_i, \nabla_{\mathbf{x}} y_j \right> \right|^2}{\| \nabla_{\mathbf{x}} y_i \|_2^2 \| \nabla_{\mathbf{x}} y_j \|_2^2} \ge 0.
    \end{align*}

\section{Proof of Theorem \ref{thm:Perf_Guar_Hetero_Half}}
\label{append:c}

  This proof is based on ideas from guarantees derived for submodular maximization under matroid constraints \cite{Filmus_2012_A_Tight}. We denote by $M = \sum_{i=1}^L M_i$ the total number of sensors to be selected. Let $\mathcal{T}^m$ denote the set selected by the JGS algorithm till iteration $m$, with the sensor $t^m$ added at the $m$-th iteration. We define $\rho_t \left( \mathcal{T}^m \right) \coloneqq \mathcal{C} \left( \mathcal{T}^m \cup t \right) - \mathcal{C} \left( \mathcal{T}^m \right)$ to be the incremental cost of adding the element $t$ to the set $\mathcal{T}^m$. Similarly, $\rho_{_\mathcal{R}} \left( \mathcal{T} \right) \coloneqq \mathcal{C} \left( \mathcal{T} \cup \mathcal{R} \right) - \mathcal{C} \left( \mathcal{T} \right)$ is the incremental cost of adding the set $\mathcal{R}$ to the set $\mathcal{T}$. This incremental cost $\rho_{_\mathcal{R}} \left( \mathcal{T} \right)$ is a submodular function as shown next.
 \begin{align}
     \rho_{_\mathcal{P}} \left( \mathcal{T} \right) &+ \rho_{_\mathcal{R}} \left( \mathcal{T} \right) \notag \\
     &= \mathcal{C} \left( \mathcal{T} \cup \mathcal{P} \right) - \mathcal{C} \left( \mathcal{T} \right) +\mathcal{C} \left( \mathcal{T} \cup \mathcal{R} \right) - \mathcal{C} \left( \mathcal{T} \right) \notag \\
     &\ge \mathcal{C} \left( \left( \mathcal{P} \cup \mathcal{R} \right) \cup \mathcal{T} \right) + \mathcal{C} \left( \left( \mathcal{P} \cap \mathcal{R} \right) \cup \mathcal{T} \right) - 2 \mathcal{C} \left( \mathcal{T} \right) \notag \\
     &= \rho_{_{\mathcal{P} \cup \mathcal{R}}} \left( \mathcal{T} \right) + \rho_{_{\mathcal{P} \cap \mathcal{R}}} \left( \mathcal{T} \right),
     \label{eq:submod_set}
 \end{align}
 where the inequality is followed by the submodularity of $\mathcal{C}$.

 The optimal set $\mathcal{T}_{OPT}$ can be divided into its subsets $\mathcal{T}_{OPT,i} \coloneqq \mathcal{T}_{OPT} \cap \mathcal{S}_i$ for $i = 1, \ldots, L$, where $\mathcal{T}_{OPT,i}$ denotes the subset of the optimal set selected from $\mathcal{S}_i$. Since $\mathcal{T}_{OPT}$ is a feasible solution of \eqref{eq:relaxed_prob}, $\left| \mathcal{T}_{OPT,i} \right| = M_i$ for $i = 1, \ldots, L$. Similarly, $\left| \mathcal{T}^\ast \cap \mathcal{S}_i \right| = M_i$ since $\mathcal{T}^\ast$ is also a feasible solution of \eqref{eq:relaxed_prob}. A one-to-one mapping can then be defined from the set $\mathcal{T}^\ast \cap \mathcal{S}_i$ to the set $\mathcal{T}_{OPT,i}$ for each $i \in \{ 1, \ldots, L\}$ as follows. All elements in $\left( \mathcal{T}^\ast \cap \mathcal{S}_i \right) \cap \mathcal{T}_{OPT,i}$, are mapped to themselves. A random bijection is constructed from the elements in $\left( \mathcal{T}^\ast \cap \mathcal{S}_i \right) \backslash \mathcal{T}_{OPT,i}$ to those in $\mathcal{T}_{OPT,i} \backslash \left( \mathcal{T}^\ast \cap \mathcal{S}_i \right)$. As a result, let this bijection map the element $v_j^i \in \left( \mathcal{T}^\ast \cap \mathcal{S}_i \right)$ to the element $o_j^i \in \mathcal{T}_{OPT,i}$, for $j = 1, \ldots, M_i$.

 Let $\mathcal{V}_i \subset \{1, \ldots, M\}$, for $i = 1, \ldots, L$, denote the index set of the iterations in which the JGS algorithm selects the sensor from $\mathcal{S}_i$. Specifically, $m \in \mathcal{V}_i$ if at the $m$-th iteration of Algorithm \ref{alg:JGS}, the sensor is selected from $\mathcal{S}_i$. Thus, we have, $\left| \mathcal{V}_i \right| = M_i$. Then, $\left\{ \mathcal{V}_i \right\}_{i=1}^L$ gives the trajectory followed by the JGS algorithm in selecting its sensors from the different sets $\mathcal{S}_1, \ldots, \mathcal{S}_L$. In particular, $\left\{ \mathcal{V}_i \right\}_{i=1}^L$ depends on $\mathcal{A}$, $\boldsymbol{\Sigma}$, and $M_i$ for $i = 1, \ldots, L$.

 For any iteration $m$ of the JGS algorithm, if $m \in \mathcal{V}_i$ for $i \in \{1, \ldots, L\}$, then at this iteration, the sensor $t^m$ is selected from the set $\left( \mathcal{T}^\ast \cap \mathcal{S}_i \right)$. Let $t^m$ be equal to $v_{j^m}^i \in \left( \mathcal{T}^\ast \cap \mathcal{S}_i \right)$ for $j^m \in \{1, \ldots, M_i\}$. Then, $\rho_{t^m} \left( \mathcal{T}^{m-1} \right) \ge \rho_{o_{j^m}^i} \left( \mathcal{T}^{m-1} \right)$, since $t^m$ is the element selected by the greedy algorithm. Thus, with $\mathcal{T}^0 = \varnothing$ denoting the initial selection, we get
 \begin{align}
     \mathcal{C} \left( \mathcal{T}^\ast \right) &= \sum_{m=1}^M \rho_{t^m} \left( \mathcal{T}^{m-1} \right) \overset{(a)}{=} \sum_{i=1}^L \sum_{m \in \mathcal{V}_i} \rho_{t^m} \left( \mathcal{T}^{m-1} \right) \notag \\
     &\overset{(b)}{\ge} \sum_{i=1}^L \sum_{m \in \mathcal{V}_i} \rho_{o_{j^m}^i} \left( \mathcal{T}^{m-1} \right) \overset{(c)}{\ge} \sum_{i=1}^L \sum_{m \in \mathcal{V}_i} \rho_{o_{j^m}^i} \left( \mathcal{T}^\ast \right) \notag \\
     &\overset{(d)}{\ge} \sum_{i=1}^L \rho_{_{\mathcal{T}_{OPT,i}}} \left( \mathcal{T}^\ast \right) \overset{(e)}{\ge}  \rho_{_{\mathcal{T}_{OPT}}} \left( \mathcal{T}^\ast \right) \notag \\
     & = \mathcal{C} \left( \mathcal{T}_{OPT} \cup \mathcal{T}^\ast \right) -  \mathcal{C} \left( \mathcal{T}^\ast \right) \overset{(f)}{\ge} \mathcal{C} \left( \mathcal{T}_{OPT} \right) -  \mathcal{C} \left( \mathcal{T}^\ast \right)
     \label{eq:thm_proof_main}
\end{align}
Here, (a) follows by rearranging the terms in the summation, (b) follows from the greediness of the JGS algorithm and the one-to-one maps constructed previously, (c) follows from the submodularity of $\mathcal{C}$, (d) and (e) follow from the submodularity of $\rho_{_\mathcal{R}} \left( \mathcal{T} \right)$ from \eqref{eq:submod_set}, and since $\mathcal{C}$ is a normalized function, and (f) follows due to the monotonicity of $\mathcal{C}$.

Finally, rearranging \eqref{eq:thm_proof_main}, Theorem \ref{thm:Perf_Guar_Hetero_Half} is proved.


\section{Proof of Theorem \ref{thm: Perf_Guar_Hetero}}
\label{append:b}

    For this proof, let $\mathcal{T}_m$ denote the set of samples selected by the algorithm till the $m$-th iteration, $c_m = \mathcal{C}\left( \mathcal{T}_m \right) - \mathcal{C}\left( \mathcal{T}_{m-1} \right)$ denote the incremental cost of the added sample at the $m$-th iteration, and $\rho_j \left( \mathcal{T}_m \right)$ denote the incremental cost of adding the element $\{j\}$ to the set $\mathcal{T}_m$. 

    In the proof, the following property of normalized, monotone, non-decreasing submodular functions $\mathcal{C}$ will be extensively used.
    \begin{equation}
        \mathcal{C}(\mathcal{T}) \le \mathcal{C}(\mathcal{S}) + \sum_{j\in (\mathcal{T}-\mathcal{S})} \rho_j(\mathcal{S}), \text{ } \forall \mathcal{S}, \mathcal{T} \in \mathcal{N}
        \label{submod_prop_3}
    \end{equation}    
    where $\rho_j(\mathcal{S}) \coloneqq \mathcal{C} \left( \mathcal{S} \cup \{ j \} \right) - \mathcal{C}(\mathcal{S})$ is the incremental cost of adding the element $j$ to the set $\mathcal{S}$.

    In order to prove Theorem \ref{thm: Perf_Guar_Hetero}, we first show the performance of the algorithm until the switching of the search space occurs at the $m_s$-th iteration and then how that affects the performance of the algorithm after the switch.
    Lemma \ref{lemma_le_ms} states the performance for $m \le m_s$, and Lemma \ref{lemma_ge_ms} states the performance for $m > m_s$, where $m$ denotes the iteration count of the algorithm.

    \begin{lemma}\label{lemma_le_ms}
        For $m \le m_s$, if $\mathcal{T}_m$ is the set selected by the algorithm till the $m$-th iteration, and $\mathcal{C}$ is a normalized non-decreasing submodular function, then
        \begin{multline}
            \mathcal{C}\left( \mathcal{T}_{m} \right) \ge \left[ 1 - \left( 1 - \frac{1}{M_1 + M_2} \right)^{m} \right] \mathcal{C} \left( \mathcal{T}_{OPT} \right), \\
            \text{ for } m = 1, 2, \ldots, m_s.
            \label{cost_ineq_le_ms}
        \end{multline}
    \end{lemma}                                                       
    \begin{proof}
    The cost function $\mathcal{C}$ is a normalized function, so that, if $\mathcal{T}_0 = \varnothing$, then, $\mathcal{C} \left( \mathcal{T}_0 \right) = 0$.
    
    The cost function $\mathcal{C}$ is a non-decreasing function, so that, $\rho_j \left( \mathcal{T}_m \right) = \mathcal{C} \left( \mathcal{T}_m \cup \{j\} \right) - \mathcal{C} \left( \mathcal{T}_m \right) \ge 0$.
    
    We use the property of submodular functions given by (\ref{submod_prop_3}), assuming the function is non-decreasing, to get,
    \begin{align}
    \mathcal{C} \left( \mathcal{T}_{OPT} \right) &\le \mathcal{C}\left( \mathcal{T}_m \right) + \sum_{j \in \left( \mathcal{T}_{OPT} - \mathcal{T}_m \right)} \rho_j \left( \mathcal{T}_m \right).
    \label{submod_basic_ineq}
    \end{align}
    
    Now, for $m \le m_s$, the algorithm chooses the best possible sample out of the entire available set of samples at each step. Thus, we get from inequality (\ref{submod_basic_ineq}),
    \begin{align}
        \mathcal{C} \left( \mathcal{T}_{OPT} \right) &\overset{\tiny{(a)}}{\le} \mathcal{C}\left( \mathcal{T}_m \right) + \sum_{j \in \left( \mathcal{T}_{OPT} - \mathcal{T}_m \right)} c_{m+1} \notag \\
        &\overset{\tiny{(b)}}{\le} \mathcal{C}\left( \mathcal{T}_m \right) + \left( M_1 + M_2 \right) c_{m+1} \notag \\
        &= \mathcal{C}\left( \mathcal{T}_m \right) + \left( M_1 + M_2 \right) \left( \mathcal{C}\left( \mathcal{T}_{m+1} \right) - \mathcal{C}\left( \mathcal{T}_m \right) \right),
        \label{submod_ineq_le_ms_1}
    \end{align}
    where step (a) follows since $c_{m+1} \ge \rho_j \left( \mathcal{T}_m \right)$ for any $j \in \left( \mathcal{S}_1 \cup \mathcal{S}_2 - \mathcal{T}_m \right)$ since each iteration of the greedy algorithm chooses the sample with the highest incremental cost. (b) follows as $\left| \mathcal{T}_{OPT} - \mathcal{T}_m \right| \le \left( M_1 + M_2 \right)$, since $\mathcal{T}_{OPT}$ contains only $\left( M_1 + M_2 \right)$ elements.
    
    Using inequality (\ref{submod_ineq_le_ms_1}), and the fact that $\mathcal{C}\left( \mathcal{T}_0 \right) = 0$, we can use inductive analysis to get,
    \begin{equation}
        \mathcal{C}\left( \mathcal{T}_{m+1} \right) \ge \left[ 1 - \left( 1 - \frac{1}{M_1 + M_2} \right)^{m+1} \right] \mathcal{C} \left( \mathcal{T}_{OPT} \right). \notag
    \end{equation}
    \end{proof}
    
    From Lemma \ref{lemma_le_ms}, we get for $m = m_s$,
    \begin{equation}
        \mathcal{C}\left( \mathcal{T}_{m_s} \right) \ge \left[ 1 - \left( 1 - \frac{1}{M_1 + M_2} \right)^{m_s} \right] \mathcal{C} \left( \mathcal{T}_{OPT} \right).
        \label{cost_ineq_ms}
    \end{equation}
    This gives the worst-case performance of the set chosen by the algorithm till the $m_s$-th iteration, that is, till just before the switch in the search space occurs. Using inequality (\ref{cost_ineq_ms}), we next find a lower bound for the performance of the algorithm for iterations $m > m_s$.
    
    \begin{lemma}\label{lemma_ge_ms}
        Assuming that the algorithm has selected all $M_1$ samples from $\mathcal{S}_1$ at the $m_s$-th iteration, for iteration $m = m_s + k$, $1 \le k \le \left( M_1 + M_2 - m_s \right)$, if $\mathcal{T}_{m_s + k}$ is the set selected by the algorithm till the $\left( m_s + k \right)$-th iteration, and $\mathcal{C}$ is a normalized non-decreasing submodular function, then
        \begin{multline}
            \mathcal{C}\left( \mathcal{T}_{m_s+k} \right) \ge \left[ 1 - \frac{M_1}{M_2} \sum_{j=0}^{k-1} \left( 1 - \frac{M_1 + 1}{M_2} \right)^j - \right. \\
            \left. \left( 1 - \frac{M_1 + 1}{M_2} \right)^k  \left( 1 - \frac{1}{M_1 + M_2} \right)^{m_s} \right] \mathcal{C} \left( \mathcal{T}_{OPT} \right), \\
            \text{for } k = 1, 2, \ldots, \left( M_1 + M_2 - m_s \right).
            \label{cost_ineq_ge_ms_1}
        \end{multline}
        
        If, on the other hand, the algorithm exhausts $M_2$ samples from $\mathcal{S}_2$ at the $m_s$-th iteration, then, under similar assumptions on the cost function,
        \begin{multline}
            \mathcal{C}\left( \mathcal{T}_{m_s+k} \right) \ge \left[ 1 - \frac{M_2}{M_1} \sum_{j=0}^{k-1} \left( 1 - \frac{M_2 + 1}{M_1} \right)^j - \right. \\
            \left. \left( 1 - \frac{M_2 + 1}{M_1} \right)^k  \left( 1 - \frac{1}{M_1 + M_2} \right)^{m_s} \right] \mathcal{C} \left( \mathcal{T}_{OPT} \right), \\
            \text{for } k = 1, 2, \ldots, \left( M_1 + M_2 - m_s \right). \label{cost_ineq_ge_ms_2}
        \end{multline}
    \end{lemma}

    \begin{proof}
    
    For $m > m_s$, we can consider as if the algorithm is starting with $\mathcal{T}_{m_s}$ as the initial set. Then, inequality (\ref{cost_ineq_ms}) is the initial cost from where the algorithm begins.
    
    Let us first assume that the algorithm has selected $M_1$ samples from $\mathcal{S}_1$ at the $m_s$-th iteration. The other case, when $M_2$ samples are exhausted from $\mathcal{S}_2$ at the $m_s$-th iteration, follows similarly.
    For $m \ge m_s$, we again use the property of submodular functions given by (\ref{submod_prop_3}), assuming the function is non-decreasing, to obtain,
    \begin{multline}
        \mathcal{C} \left( \mathcal{T}_{OPT} \right) \le \mathcal{C}\left( \mathcal{T}_m \right) + \sum_{j \in \left( \mathcal{T}_{OPT} - \mathcal{T}_m \right) \cap \mathcal{S}_1} \rho_j \left( \mathcal{T}_m \right) + \\
        \sum_{j \in \left( \mathcal{T}_{OPT} - \mathcal{T}_m \right) \cap \mathcal{S}_2} \rho_j \left( \mathcal{T}_m \right).
    \end{multline}
    
    Now, from the $\left( m_s+1\right)$-th iteration, i.e., for $m > m_s$ the algorithm selects the samples from the set $\left( \mathcal{S}_2 - \mathcal{T}_m \right)$ only. It can be seen from Fig.~(\ref{fig:venn_diag}) that $\left( \mathcal{T}_{OPT} - \mathcal{T}_m \right) \cap \mathcal{S}_2 = \left( \mathcal{T}_{OPT} - \mathcal{T}_m \right) \cap \left( \mathcal{S}_2 - \mathcal{T}_m \right)$. Thus, we get,
    \begin{align}
        \mathcal{C} \left( \mathcal{T}_{OPT} \right) &\overset{\tiny{(a)}}{\le} \mathcal{C}\left( \mathcal{T}_m \right) + \sum_{j \in \left( \mathcal{T}_{OPT} - \mathcal{T}_m \right) \cap \left( \mathcal{S}_2 - \mathcal{T}_m \right)} c_{m+1} \notag \\
        & + \sum_{j \in \left( \mathcal{T}_{OPT} - \mathcal{T}_m \right) \cap \mathcal{S}_1} \rho_j \left( \mathcal{T}_m \right) & \notag \\
        &\overset{\tiny{(b)}}{\le} \mathcal{C}\left( \mathcal{T}_m \right) + M_2 c_{m+1} + \sum_{j \in \left( \mathcal{T}_{OPT} - \mathcal{T}_m \right) \cap \mathcal{S}_1} \rho_j \left( \mathcal{T}_m \right) \notag \\
        &\overset{\tiny{(c)}}{\le} \mathcal{C}\left( \mathcal{T}_m \right) + M_2 c_{m+1} \notag \\
        & + \sum_{j \in \left( \mathcal{T}_{OPT} - \mathcal{T}_m \right) \cap \mathcal{S}_1} \rho_j \left( \mathcal{T}_{m_s - 1} \right) \notag
    \end{align}
    \begin{align}
        \mathcal{C} \left( \mathcal{T}_{OPT} \right) 
        &\overset{\tiny{(d)}}{\le} \mathcal{C}\left( \mathcal{T}_m \right) + M_2 c_{m+1} + \sum_{j \in \left( \mathcal{T}_{OPT} - \mathcal{T}_m \right) \cap \mathcal{S}_1} c_{m_s} \notag \\
        &\overset{\tiny{(e)}}{\le} \mathcal{C}\left( \mathcal{T}_m \right) + M_2 c_{m+1} + M_1 c_{m_s} \notag \\
        &\overset{\tiny{(f)}}{\le} M_2 \mathcal{C}\left( \mathcal{T}_{m+1} \right) - \left( M_2 - M_1 - 1 \right) \mathcal{C}\left( \mathcal{T}_m \right).
        \label{submod_ineq_ge_ms_1}
    \end{align}
    Here, 
    \begin{itemize}[label=$-$]
        \item (a) follows since the greedy algorithm chooses from the set $\left( \mathcal{S}_2 - \mathcal{T}_m \right)$ at the $m$-th iteration, making the incremental cost $c_{m+1} \ge \rho_j \left( \mathcal{T}_m \right)$ for any $j \in \left( \mathcal{T}_{OPT} - \mathcal{T}_m \right) \cap \left( \mathcal{S}_2 - \mathcal{T}_m \right) \subseteq \left( \mathcal{S}_2 - \mathcal{T}_m \right)$.
        \item (b) follows as $\left( \mathcal{T}_{OPT} - \mathcal{T}_m \right) \cap \mathcal{S}_2 = \left( \mathcal{T}_{OPT} - \mathcal{T}_m \right) \cap \left( \mathcal{S}_2 - \mathcal{T}_m \right) \subseteq \left( \mathcal{T}_{OPT} \cap \mathcal{S}_2 \right)$, thus making $\left| \left( \mathcal{T}_{OPT} - \mathcal{T}_m \right) \cap \left( \mathcal{S}_2 - \mathcal{T}_m \right) \right| \le \left| \left( \mathcal{T}_{OPT} \cap \mathcal{S}_2 \right) \right| = M_2$.
        \item (c) follows from the definition of submodular functions, since $\left( \mathcal{T}_{m_s - 1} \right) \subseteq \left( \mathcal{T}_m \right)$.
        \item (d) is true since at the $m_s$-th iteration of the greedy algorithm, the incremental cost $c_{m_s} \ge \rho_j \left( \mathcal{T}_{m_s - 1} \right)$ for any $j \notin \mathcal{T}_{m_s -1}$. Now, as $\mathcal{T}_{m_s -1} \subseteq \mathcal{T}_m$, then $\left( \left( \mathcal{T}_{OPT} - \mathcal{T}_m \right) \cap \mathcal{S}_1 \right) \subseteq \left( \left( \mathcal{T}_{OPT} - \mathcal{T}_{m_s - 1} \right) \cap \mathcal{S}_1 \right)$. This implies, if $j \in \left( \mathcal{T}_{OPT} - \mathcal{T}_m \right) \cap \mathcal{S}_1$, then $j \notin \mathcal{T}_{m_s - 1}$.
        \item (e) is true since $\left( \left( \mathcal{T}_{OPT} - \mathcal{T}_m \right) \cap \mathcal{S}_1 \right) \subseteq \left( \mathcal{T}_{OPT} \cap \mathcal{S}_1 \right)$, thus making $\left| \left( \mathcal{T}_{OPT} - \mathcal{T}_m \right) \cap \mathcal{S}_1 \right| \le \left| \mathcal{T}_{OPT}  \cap \mathcal{S}_1 \right| = M_1$.
        \item (f) follows from the monotone increasing property of the given cost. Thus, $\mathcal{C}\left( \mathcal{T}_{m_s} \right) \le \mathcal{C}\left( \mathcal{T}_{m} \right)$.
    \end{itemize}
    
    \begin{figure}[t]
        \centering
        \includegraphics[width=0.48\textwidth]{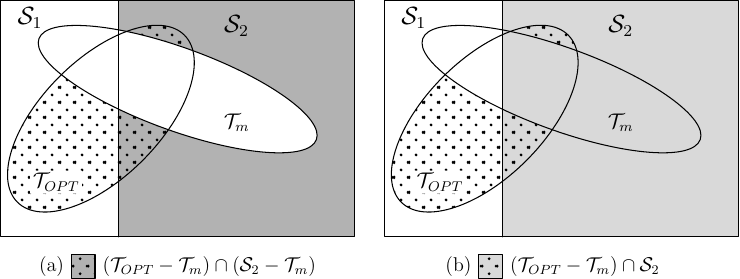}
        \caption{Venn diagram to show $\left( \mathcal{T}_{OPT} - \mathcal{T}_m \right) \cap \left( \mathcal{S}_2 - \mathcal{T}_m \right) = \left( \mathcal{T}_{OPT} - \mathcal{T}_m \right) \cap \mathcal{S}_2$}
        \label{fig:venn_diag}
        \vspace{-0.1in}
    \end{figure}
    
    Now, rearranging the inequality (\ref{submod_ineq_ge_ms_1}), we get,
    \begin{equation}
        \mathcal{C}\left( \mathcal{T}_{m+1} \right) \ge \frac{1}{M_2} \mathcal{C} \left( \mathcal{T}_{OPT} \right) + \left( 1 - \frac{M_1 + 1}{M_2} \right) \mathcal{C}\left( \mathcal{T}_m \right).
        \label{submod_ineq_ge_ms_2}
    \end{equation}
    
    Using inequalities (\ref{submod_ineq_ge_ms_1}) and (\ref{cost_ineq_ms}), we get, for $k = 1, 2, \ldots, \left( M_1 + M_2 - m_s \right)$,
    \begin{multline}
        \mathcal{C}\left( \mathcal{T}_{m_s+k} \right) \ge \left[ 1 - \frac{M_1}{M_2} \sum_{j=0}^{k-1} \left( 1 - \frac{M_1 + 1}{M_2} \right)^j \right.\\
        \left. - \left( 1 - \frac{M_1 + 1}{M_2} \right)^k \left( 1 - \frac{1}{M_1 + M_2} \right)^{m_s} \right] \mathcal{C} \left( \mathcal{T}_{OPT} \right). \notag
    \end{multline}
    
    If the algorithm selects $M_2$ samples from $\mathcal{S}_2$ first at the $m_s$-th iteration, the roles of $M_1$ and $M_2$ get interchanged in the above steps. Similarly, the roles of $\mathcal{S}_1$ and $\mathcal{S}_2$ get interchanged. Using the above analysis, we arrive at inequality (\ref{cost_ineq_ge_ms_1}).
    \end{proof}

    Theorem \ref{thm: Perf_Guar_Hetero} follows directly from Lemma \ref{lemma_ge_ms}, by replacing $m_s + k = M_1 + M_2$, that is, by evaluating the worst-case performance of the algorithm from (\ref{cost_ineq_ge_ms_1}) for the last iteration.

}

\bibliographystyle{IEEEtran}
\bibliography{IEEEabrv,ref_abbr}

\end{document}